\documentclass[11pt]{article}
\usepackage{amsfonts,amsmath,amsthm,amssymb,dsfont, etex}

\usepackage{tikz}
\usetikzlibrary{arrows,automata}
\usepackage{comment, times}
\usepackage{fullpage}
\usepackage[margin=1in]{geometry}
\usepackage{soul,color}
\usepackage{graphicx,float,wrapfig}
\usepackage{mathrsfs}
\usepackage[usenames,dvipsnames]{pstricks}
\usepackage{epsfig}
\usepackage[linesnumbered,boxed,ruled,vlined]{algorithm2e}
\usetikzlibrary{positioning}

\usepackage{authblk}
\usepackage{url}

\usepackage{tabu}
\usepackage{enumitem}

\usepackage{multirow}
\usepackage{array}
\usepackage{tabularx}
\newcolumntype{L}[1]{>{\raggedright\let\newline\\\arraybackslash\hspace{0pt}}m{#1}}
\newcolumntype{C}[1]{>{\centering\let\newline\\\arraybackslash\hspace{0pt}}m{#1}}
\newcolumntype{R}[1]{>{\raggedleft\let\newline\\\arraybackslash\hspace{0pt}}m{#1}}
\usepackage{diagbox}

\newtheorem{theo}{Theorem}[section]

\newtheorem{lemma}[theo]{Lemma}

\newtheorem{cor}[theo]{Corollary}
\theoremstyle{definition}
\newtheorem{defn}[theo]{Definition}
\theoremstyle{plain}

\newtheorem{fact}{Fact}
\newtheorem{conj}{Conjecture}[section]

\everymath{\displaystyle}

\renewcommand{\d}{\mathrm{d}}

\newcommand{\BQP}{\mathsf{BQP}}

\newcommand{\BPP}{\mathsf{BPP}}

\newcommand{\PH}{\mathsf{PH}}
\newcommand{\PP}{\mathsf{PP}}

\newcommand{\NP}{\mathsf{NP}}

\newcommand{\PostBQP}{\mathsf{PostBQP}}

\newcommand{\eps}{\epsilon}

\newcommand{\DQC}{\mathsf{DQC}}
\newcommand{\IQP}{\mathsf{IQP}}
\newcommand{\Gap}{\mathsf{Gap}}

\renewcommand{\epsilon}{\varepsilon}

\def\eq#1{Eq. \!\eqref{#1}}

\def\braket#1#2{\langle #1 | #2 \rangle}

\def\set#1{\{#1\}}

\def\mat#1{\left(\begin{matrix} #1 \end{matrix}\right)}
\def\PWEAK{\mathsf{PWEAK}}

\def\cP{\mathcal P}
\def\cQ{\mathcal Q}

\def\bbC{\mathbb C}

\def\bbR{\mathbb R}

\def\bbZ{\mathbb Z}

\setlist[enumerate]{itemsep=-1mm}

\def\ShowAuthNotes{1}
\ifnum\ShowAuthNotes=1
\newcommand{\authnote}[2]{\ \\ \textcolor{red}{\parbox{0.9\linewidth}{[{\footnotesize {\bf #1:} { {#2}}}]}}\newline}
\else
\newcommand{\authnote}[2]{}
\fi

\let\svfootnoterule\footnoterule
\renewcommand\footnoterule{\vfill\svfootnoterule}
\input{QCircuit.tex}

\begin{document}

\title{Complexity Classification of Conjugated Clifford Circuits}

\author[1]{Adam Bouland\thanks{adam@csail.mit.edu}}
\author[2,3]{Joseph F. Fitzsimons\thanks{joe.fitzsimons@nus.edu.sg}}
\author[4]{Dax Enshan Koh\thanks{daxkoh@mit.edu}}
\affil[1]{\small Department of Electrical Engineering and Computer Sciences, UC Berkeley, Berkeley, CA, USA}
\affil[2]{\small Singapore University of Technology and Design, 8 Somapah Road, Singapore 487372}
\affil[3]{\small Centre for Quantum Technologies, National University of Singapore, 3 Science Drive 2, Singapore 117543}
\affil[4]{\small Department of Mathematics, Massachusetts Institute of Technology, Cambridge, Massachusetts 02139, USA}

\date{}


\clearpage\maketitle

\begin{abstract}
Clifford circuits --  i.e.\ circuits composed of only CNOT, Hadamard, and $\pi/4$ phase gates -- play a central role in the study of quantum computation.
However, their computational power is limited: a well-known result of Gottesman and Knill \cite{gottesmanknill} states that Clifford circuits are efficiently classically simulable. 
We show that in contrast, ``conjugated Clifford circuits" (CCCs) -- where one additionally conjugates every qubit by the same one-qubit gate $U$ -- can perform hard sampling tasks. 
In particular, we fully classify the computational power of CCCs by showing that essentially any non-Clifford conjugating unitary $U$  can give rise to sampling tasks which cannot be efficiently classically simulated to constant multiplicative error, unless the polynomial hierarchy collapses. 
Furthermore, by standard techniques, this hardness result can be extended to allow for the more realistic model of constant additive error, under a plausible complexity-theoretic conjecture.
This work can be seen as progress towards classifying the computational power of all restricted quantum gate sets.

\end{abstract}

\section{Introduction}

Quantum computers hold the promise of efficiently solving certain problems, such as factoring integers \cite{shor1999polynomial}, which are believed to be intractable for classical computers.
However, experimentally implementing many of these quantum algorithms is very difficult.
For instance, the largest number factored to date using Shor's algorithm is 21 \cite{martin2012experimental}. These considerations have led to an intense interest in algorithms which can be more easily implemented with near-term quantum devices, as well as a corresponding interest in the difficulty of these computational tasks for classical computers. Some prominent examples of such models are constant-depth quantum circuits \cite{terhaldivincenzo,bravyi2017quantum} and non-adaptive linear optics \cite{aaronsonboson}.

In many of these constructions, one can show that these ``weak" quantum devices can perform \emph{sampling} tasks which cannot be efficiently simulated classically, even though they may not be known to be capable of performing difficult \emph{decision} tasks. Such arguments were first put forth by Bremner, Jozsa, and Shepherd \cite{BJS2010} and Aaronson and Arkhipov \cite{aaronsonboson}, who showed that exactly simulating the sampling tasks performed by weak devices is impossible assuming the polynomial hierarchy is infinite. 
The proofs of these results use the fact that the output probabilities of quantum circuits can be very difficult to compute -- in fact they can be $\Gap\mathsf{P}$-hard and therefore $\#\mathsf{P}$-hard to approximate. In contrast the output probabilities of classical samplers can be approximated in $\BPP^\NP$ by Stockmeyer's approximate counting theorem \cite{stockmeyer}, and therefore lie in the polynomial hierarchy.
Hence a classical simulation of these circuits would collapse the polynomial hierarchy to the third level by Toda's Theorem \cite{Toda91}.
Similar hardness results have been shown for many other models of quantum computation \cite{terhaldivincenzo,morimae2014hardness,feffermanfourier,boulandccc2016,farhi2016quantum,boixo2016characterizing,bremner2016achieving,bremner2017upcoming}.

A curious feature of many of these ``weak" models of quantum computation is that they can be implemented using non-universal gate sets.
That is, despite being able to perform sampling problems which appear to be outside of $\BPP$, these models are not themselves known to be capable of universal quantum computation. 
In short these models of quantum computation seem to be ``quantum-intermediate" between $\BPP$ and $\BQP$, analogous to the $\NP$-intermediate problems which are guaranteed to exist by Ladner's theorem \cite{Ladner}. From the standpoint of computational complexity, it is therefore natural to study this intermediate space between $\BPP$ and $\BQP$, and to classify its properties. 

One natural way to explore the space between $\BPP$ and $\BQP$ is to classify the power of all possible quantum gate sets over qubits.
The Solovay-Kitaev Theorem states that all universal quantum gate sets, i.e.\ those which densely generate the full unitary group, have equivalent computational power \cite{dawson2006solovay}. Therefore the interesting gates sets to classify are those which are non-universal.
However just because a gate set is non-universal does not imply it is weaker than $\BQP$ -- in fact some non-universal gates are known to be capable of universality in an ``encoded" sense, and therefore have the same computational power as $\BQP$ \cite{jozsa2008matchgates}.
Other non-universal gate sets are efficiently classically simulable \cite{gottesmanknill}, while others seem to lie ``between $\BPP$ and $\BQP$" in that they are believed to be neither universal for $\BQP$ nor efficiently classically simulable \cite{BJS2010}.
It is a natural open problem to fully classify all restricted gate sets into these categories according to their computational complexity.

This is a challenging problem, and to date there has only been partial progress towards this classification.
One immediate difficulty in approaching this problem is that there is not a known classification of all possible non-universal gate sets. In particular this would require classifying the discrete subgroups of $SU(2^n)$ for all $n\in\mathbb{N}$, which to date has only been solved for $n\leq 2$ \cite{Hanany2001}.
Therefore existing results have characterized the power of modifications of known intermediate gate sets, such as commuting circuits and linear optical elements \cite{bouland2014generation,boulandccc2016,oszmaniec2017universal}.
Others works have classified the classical subsets of gates \cite{reversibleclassification,cliffordclassification}, or else given sufficient criteria for universality so as to rule out the existence of certain intermediate families \cite{sawicki2016criteria}.
A complete classification of this space ``between $\BPP$ and $\BQP$" would require a major improvement in our understanding of universality as well as the types of computational hardness possible between $\BPP$ and $\BQP$.

One well-known example of a non-universal family of quantum gates is the Clifford group. Clifford circuits -- i.e. circuits composed of merely CNOT, Hadamard and Phase gates -- are a discrete subgroup of quantum gates which play an important role in quantum error correction \cite{gottesman1997thesis,bravyi2005universal}, measurement-based quantum computing \cite{raussendorf2001one,raussendorf2003measurement,briegel2009measurement}, and randomized benchmarking \cite{magesan2011scalable}. 
However a well-known result of Gottesman and Knill states that circuits composed of Clifford elements are efficiently classically simulable \cite{gottesmanknill,aaronson2004clifford}. 
That is, suppose one begins in the state $\ket{0}^{\otimes n}$, applies polynomially many gates from the set $\mathrm{CNOT}, \mathrm{H}, \mathrm{S}$, then measures in the computational basis. Then the Gottesman-Knill theorem states that one can compute the probability a string $y$ is output by such a circuit in classical polynomial time. One can also sample from the same probability distribution on strings as this circuit as well.
A key part of the proof of this result is that the quantum state at intermediate stages of the circuit is always a ``stabilizer state" -- i.e. the state is uniquely described by its set of stabilizers in the Pauli group -- and therefore has a compact representation.
Therefore the Clifford group is incapable of universal quantum computation (assuming $\BPP\neq \BQP$).

In this work, we will study the power of a related family of non-universal gates, known as \emph{Conjugated Clifford gates}, which we introduce below. These gates are non-universal by construction, but not known to be efficiently classically simulable either. Our main result will be to \emph{fully classify} the computational power of this family of intermediate gate sets.

\subsection{Our results}

This paper considers a new ``weak" model of quantum computation which we call ``conjugated Clifford circuits" (CCCs). 
In this model, we consider the power of quantum circuits which begin in the state $\ket{0}^{\otimes n}$, and then apply gates from the set $(U^\dagger \otimes U^\dagger)(\mathrm{CNOT})(U \otimes U), U^\dag HU,U^\dag SU$ where $U$ is a fixed one-qubit gate. In other words, we consider the power of Clifford circuits which are conjugated by an identical one-qubit gate $U$ on each qubit. These gates manifestly perform a discrete subset of unitaries so this gate set is clearly not universal.

Although this transformation preserves the non-universality of the Clifford group, it is unclear if it preserves its computational power.
The presence of generic conjugating unitaries (even the same $U$ on each qubit, as in this model) breaks the Gottesman-Knill simulation algorithm \cite{gottesmanknill}, as the inputs and outputs of the circuit are not stabilizer states/measurements. Hence the intermediate states of the circuit are no longer efficiently representable by the stabilizer formalism. This, combined with prior results showing hardness for other modified versions of Clifford circuits \cite{jozsa2014clifford,koh2015further}, leads one to suspect that CCCs may not be efficiently classically simulable.
However prior to this work no hardness results were known for this model.

In this work, we confirm this intuition and provide two results in this direction.
First, we provide a \emph{complete classification} of the power of CCCs according to the choice of $U$.
We do this by showing that \emph{any} $U$ which is not efficiently classically simulable by the Gottesman-Knill theorem suffices to perform hard sampling problems with CCCs\footnote{More precisely, we show that any $U$ that cannot be written as a Clifford times a $Z$-rotation suffices to perform hard sampling problems with CCCs. See Theorem \ref{cor:classificationWeak} for the exact statement.}. That is, for generic $U$, CCCs cannot be efficiently classically simulated to constant multiplicative error by a classical computer unless the polynomial hierarchy collapses. 
This result can be seen as progress towards classifying the computational complexity of restricted gate sets. 
Indeed, given a non-universal gate set $G$, a natural question is to classify the power of $G$ when conjugated by the same one-qubit unitariy $U$ on each qubit, as this transformation preserves non-universality.
Our work resolves this question for one of the most prominent examples of non-universal gate sets, namely the Clifford group.
As few examples of non-universal gate sets are known\footnote{The only examples to our knowledge are matchgates, Clifford gates, diagonal gates, and subsets thereof.}, this closes one of the major gaps in our understanding of intermediate gate sets. Of course this does not complete the complexity classification of all gate sets, as there is no known classification of all possible non-universal gate sets. However it does make progress towards this goal.

Second, we show that under an additional complexity-theoretic conjecture, classical computers cannot efficiently simulate CCCs to constant error in total variation distance. This is a more experimentally achievable model of error for noisy quantum computations.
The proof of this result uses standard techniques introduced by Aaronson and Arkhipov \cite{aaronsonboson}, which have also been used in other models \cite{feffermanfourier,bremner2016average,boixo2016characterizing,bremner2016achieving,morimae2017hardness,bremner2017upcoming}. 

This second result is interesting for two reasons.
First, it means our results may have relevance to the empirical demonstration of quantum advantage (sometimes referred to as ``quantum supremacy") \cite{preskill2012quantum,boixo2016characterizing,aaronson2016complexity}, as our results are robust to noise.
Second, from the perspective of computational complexity, it gives yet another conjecture upon which one can base the supremacy of noisy quantum devices.
As is the case with other quantum supremacy proposals \cite{aaronsonboson,feffermanfourier, bremner2016average,morimae2017hardness,bremner2017upcoming}, in order to show that simulation of CCCs to additive error still collapses the polynomial hierarchy, we need an additional conjecture stating that the output probabilities of these circuits are hard to approximate on average.
Our conjecture essentially states that for most Clifford circuits $V$ and most one-qubit unitaries $U$, it is $\#\mathsf{P}$-hard to approximate a constant fraction of the output probabilities of the CCC $U^{\otimes n} V (U^\dagger)^{\otimes n}$ to constant multiplicative error.
We prove that this conjecture is true in the worst case -- in fact, for all non-Clifford $U$, there exists a $V$ such that some outputs are $\#\mathsf{P}$-hard to compute to multiplicative error.
However, it remains open to extend this hardness result to the average case, as is the case with other supremacy proposals as well \cite{aaronsonboson,feffermanfourier, bremner2016average,morimae2017hardness,bremner2017upcoming}.
To the best of our knowledge our conjecture is independent of the conjectures used to establish other quantum advantage results such as boson sampling \cite{aaronsonboson}, Fourier sampling \cite{feffermanfourier} or IQP \cite{bremner2016average,bremner2016achieving}.
Therefore our results can be seen as establishing an alternative basis for belief in the advantage of noisy quantum devices over classical computation.

One final motivation for this work is that CCCs might admit a simpler fault-tolerant implementation than universal quantum computing, which we conjecture to be the case. It is well-known that many stabilizer error-correcting codes, such as the 5-qubit and 7-qubit codes \cite{laflamme1996perfect,5qubitcode,steane1996multiple}, admit transversal Clifford operations \cite{gottesman1997thesis}. That is, performing fault-tolerant Clifford operations on the encoded logical qubits can be done in a very simple manner -- by simply performing the corresponding Clifford operation on the physical qubits. This is manifestly fault-tolerant, in that an error on one physical qubit does not ``spread" to more than 1 qubit when applying the gate. 
In contrast, performing non-Clifford operations fault-tolerantly on such codes requires substantially larger (and non-transversal) circuits -- and therefore the non-transversal operations are often the most resource intensive. 
The challenge in fault-tolerantly implementing CCCs therefore lies in performing the initial state preparation and measurement.
Initial preparation of non-stabilizer states in these codes is equivalent to the challenge of producing magic states, which are already known to boost Clifford circuits to universality using adaptive Clifford circuits \cite{bravyi2005universal, bravyi2012magic} (in contrast our construction would only need non-adaptive Clifford circuits with magic states). 
Likewise, measuring in a non-Clifford basis would require performing non-Clifford one-qubit gates prior to fault-tolerant measurement in the computational basis.
Therefore the state preparation/measurement would be the challenging part of fault-tolerantly implementing CCCs in codes with transversal Cliffords.
It remains open if there exists a code with transversal conjugated Cliffords\footnote{Of course one can always ``rotate" a code with transversal Clifford operations to obtain a code with transversal conjugated Cliffords. If the code previously had logical states $\ket{0}_L,\ket{1}_L$, then by setting the states $\ket{0}'_L = U^\dagger_L \ket{0}_L$ and $\ket{1}'_L = U^\dagger_L \ket{1}_L$, one obtains a code in which the conjugated Clifford gates (conjugated by $U$) are transversal. However having the ability to efficiently fault-tolerantly prepare $\ket{0}_L$ in the old code does not imply the same ability to prepare $\ket{0}'_L$ in the new code.} and easy preparation and measurement in the required basis.
Such a code would not be ruled out by the Eastin-Knill Theorem  \cite{eastin2009restrictions}, which states that the set of transversal gates must be discrete for all codes which correct arbitrary one qubit errors.
Of course this is not the main motivation for exploring the power of this model -- which is primarily to classify the space between $\BPP$ and $\BQP$ -- but an easier fault-tolerant implementation could be an unexpected bonus of our results.

\subsection{Proof Techniques}

To prove these results, we use several different techniques.

\subsubsection{Proof Techniques: classification of exact sampling hardness}
To prove exact (or multiplicative) sampling hardness for CCCs for essentially all non-Clifford $U$, we use the notion of postselection introduced by Aaronson \cite{aaronson2005quantum}. Postselection is the (non-physical) ability to discard all runs of the computation which do not achieve some particular outcomes.
Our proof works by showing that postselecting such circuits allows them to perform universal quantum computation. Hardness then follows from known techniques \cite{aaronson2005quantum, BJS2010,aaronsonboson}. 

One technical subtlety that we face in this proof, which is not present in other results, is that our postselected gadgets perform operations which are not closed under inversion. This means one cannot use the Solovay-Kitaev theorem to change quantum gate sets \cite{dawson2006solovay}.
This is a necessary step in the proof that $\PostBQP=\PP$ \cite{aaronson2005quantum}, which is a key part of the hardness proof (see \cite{boulandccc2016}). Fortunately, it turns out that we can get away without inverses due to a recent inverse-free Solovay-Kitaev theorem of Sardharwalla \emph{et al.} \cite{sardharwalla2016universal}, which removes the needs for inverses if the gate set contains the Paulis. Our result would have been much more difficult to obtain without this prior result.
To our knowledge this is the first application of their result to structural complexity. 

A further difficulty in the classification proof is that the postselection gadgets we derive do not work for all non-Clifford $U$. In general, most postselection gadgets give rise to non-unitary operations, and for technical reasons we need to work with unitary postselection gadgets to apply the results of \cite{sardharwalla2016universal}. Therefore, we instead use several different gadgets which cover different portions of the parameter space of $U$'s. Our initial proof of this fact used a total of seven postselection gadgets found by hand. We later simplified this to two postselection gadgets by conducting a brute-force search for suitable gadgets using Christopher Granade and Ben Criger's QuaEC package \cite{cgranade}. We include this simplified proof in this writeup.

A final difficulty that one often faces with postselected universality proofs is that one must show that the postselection gadgets boost the original gate set to universality. In general this is a nontrivial task; there is no simple test of whether a gate set is universal, though some sufficient (but not necessary) criteria are known \cite{sawicki2016criteria}. 
Prior gate set classification theorems have solved this universality problem using representation theory \cite{bouland2014generation,sawicki2016criteria} or Lie theory \cite{boulandccc2016,oszmaniec2017universal}.
However, in our work we are able to make use of a powerful fact: namely that the Clifford group plus \emph{any} non-Clifford unitary is universal.
This follows from results of Nebe, Rains and Sloane \cite{nebe2001invariants,nebe2006self,cliffordstackexchange} classifying the invariants of the Clifford group\footnote{However we note that in our proofs we will only use the fact that the Clifford group plus any non-Clifford element is universal on a qubit. This version of the theorem admits a direct proof using the representation theory of $SU(2)$.}.
As a result our postselected universality proofs are much simpler than in other gate set classification theorems.

\subsubsection{Proof techniques: additive error}

To prove hardness of simulation to additive error, we follow the techniques of \cite{aaronsonboson,bremner2016average,feffermanfourier,morimae2017hardness}. 
In these works, to show hardness of sampling from some probability distribution with additive error, one combines three different ingredients. 
The first is anti-concentration -- showing that for these circuits, the output probabilities in some large set $T$ are somewhat large. 
Second, one uses Markov's inequality to argue that, since the simulation error sums to $\epsilon$, on some other large set of output probabilities $S$, the error must be below a constant multiple of the average. 
If $S$ and $T$ are both large, they must have some intersection -- and on this intersection $S\cap T$, the imagined classical simulation is not only a simulation to additive error, but also to multiplicative error as well (since the output probability in question is above some minimum). 
Therefore a simulation to some amount $\epsilon$ of additive error implies a multiplicative simulation to the output probabilities on a constant fraction of the outputs. 
The impossibility of such a simulation is then obtained by assuming that computing these output probabilities is multiplicatively hard on average. In particular, one assumes that it is a $\#\mathsf{P}$-hard task to compute the output probability on $|S\cap T|/2^n$ -fraction of the outputs. This leads to a collapse of the polynomial hierarchy by known techniques \cite{aaronsonboson, BJS2010}.

We follow this technique to show hardness of sampling with additive error. 
In our case, the anticoncentration theorem follows from the fact that the Clifford group is a ``2-design" \cite{webbclifford,zhuclifford} -- i.e. a random Clifford circuit behaves equivalently to a random unitary up to its second moment -- and therefore must anticoncentrate, as a random unitary does (the fact that unitary designs anticoncentrate was also shown independently by several groups \cite{hangleiter2017anti, bremner2017upcoming,saeed2018upcoming}). 
This is similar to the hardness results for $\IQP$ \cite{bremner2016average} and $\DQC 1$ \cite{morimae2017hardness}, in which the authors also prove their corresponding anticoncentration theorems. In contrast it is open to prove the anticoncentration theorem used for Boson Sampling and Fourier Sampling \cite{aaronsonboson,feffermanfourier}, though these models have other complexity-theoretic advantages\footnote{For instance, for these models it is known to be $\#\mathsf{P}$-hard to \emph{exactly} compute most output probabilities of their corresponding circuit. This is a necessary but not sufficient condition for the supremacy conjectures to be true, which require it to be $\#\mathsf{P}$-hard to \emph{approximately} compute most output probabilities of their corresponding circuit.}. 
Therefore the only assumption needed is the hardness-on-average assumption. 
We also show that our hardness assumption is true for worst-case inputs. This result follows from combining known facts about $\BQP$ with the classification theorem for exact sampling hardness.

\subsection{Relation to other works on modified Clifford circuits}

While we previously discussed the relation of our results to prior work on gate set classification and sampling problems, here we compare our results to prior work on Clifford circuits. 
We are not the first to consider the power of modified Clifford circuits. 
Jozsa and van den Nest \cite{jozsa2014clifford} and Koh \cite{koh2015further}, categorized the computational power of a number of modified versions of Clifford circuits. 
The closest related result is the statement in \cite{jozsa2014clifford} that if the input state to a Clifford circuit is allowed to be an arbitrary tensor product of one-qubit states, then such circuits cannot be efficiently classically simulated unless the polynomial hierarchy collapses. 
Their hardness result uses states of the form $\ket{0}^{\otimes n/2} \ket{\alpha}^{\otimes n/2}$, where $\ket{\alpha} = \cos(\pi/8) \ket{0}+i\sin(\pi/8)\ket{1}$ is a magic state. 
They achieve postselected hardness via the use of magic states to perform T gates, using a well-known construction (see e.g. \cite{bravyi2005universal}). 
So in the \cite{jozsa2014clifford} construction there are different input states on different qubits. 
In contrast, our result requires the same input state on every qubit -- as well as measurement in that basis at the end of the circuit. 
This ensures our modified circuit can be interpreted as the action of a discrete gate set, and therefore our result has relevance for the classification of the power of non-universal gate sets.

\section{Preliminaries}

We denote the single-qubit \textit{Pauli matrices} by $X = \sigma_x = \mat{0 & 1 \\ 1 & 0}$, $Y=\sigma_y = \mat{0 & -i \\ i & 0}$, $Z= \sigma_z = \mat{1 & 0 \\ 0 & -1}$, and $I = \mat{1 & 0 \\ 0 & 1}$. The $\pm 1$-eigenstates of $Z$ are denoted by $\ket 0$ and $\ket 1$ respectively.
The \textit{rotation operator} about an axis $t \in \{x,y,z\}$ with an angle $\theta \in [0,2\pi)$ is
\begin{equation}
R_t (\theta) = e^{-i \theta \sigma_t/2} = \cos(\theta/2) I - i \sin(\theta/2) \sigma_t.
\end{equation}
We will use the fact that any single-qubit unitary operator $U$ can be written as
\begin{equation} \label{eq:Udecomposition}
U = e^{i \alpha} R_z(\phi) R_x(\theta) R_z(\lambda),
\end{equation}
where $\alpha, \phi, \theta, \lambda \in [0,2\pi)$ \cite{nielsenchuang}.

For linear operators $A$ and $B$, we write $A \propto B$ to mean that there exists $\alpha \in \bbC\backslash\{0\}$ such that $A = \alpha B$. For linear operators, vectors or complex numbers $a$ and $b$, we write $a \sim b$ to mean that $a$ and $b$ differ only by a global phase, i.e. there exists $\theta \in [0,2\pi)$ such that $a = e^{i\theta} b$. For any subset $S\subseteq \bbR$ and $k \in \bbR$, we write $kS$ to refer to the set $\{kn:n\in S \}$. For example, $k\bbZ = \{kn:n\in \bbZ\}$. We denote the set of odd integers by $\bbZ_{odd}$. We denote the complement of a set $S$ by $S^c$.

\subsection{Clifford circuits and conjugated Clifford circuits}

The $n$-qubit \textit{Pauli group} $\mathcal P_n$ is the set of all operators of the form $i^k P_1 \otimes \ldots \otimes P_n$, where $k \in \{0,1,2,3\}$ and each $P_j$ is a Pauli matrix. The $n$-qubit \textit{Clifford group} is the normalizer of $\mathcal P_n$ in the $n$-qubit unitary group  $\mathcal U_n$, i.e. $\mathcal C_n = \{U \in \mathcal U_n: U \mathcal P_n U^\dag = \mathcal P_n\}$. 

The elements of the Clifford group, called \textit{Clifford operations}, have an alternative characterization: an operation is a Clifford operation if and only if it can be written as a circuit comprising the following gates, called \textit{basic Clifford gates}: \textit{Hadamard}, $\pi/4$ \textit{phase}, and \textit{controlled-NOT} gates, whose matrix representations in the computational basis are  
$$H=\frac{1}{\sqrt 2}\mat{1 & 1 \\ 1 & -1},\quad  S=\mat{1 & 0 \\ 0 & i},\quad \mbox{and}\quad  \mathrm{CNOT} = \mat{1 & 0 & 0 & 0 \\  0 & 1 & 0 & 0 \\ 0 & 0 & 0 & 1 \\ 0 & 0 & 1 & 0}$$ respectively. An example of a non-Clifford gate is the $T$ gate, whose matrix representation is given by $T=\mat{1 & 0 \\ 0 & e^{i \pi/4}}$. We denote the group generated by the single-qubit Clifford gates by $\langle S, H \rangle$.

We will make use of the following fact about Clifford operations.
\begin{fact} \label{fact:RzClifford}
$R_z(\phi)$ is a Clifford operation if and only if $\phi \in \tfrac \pi 2 \bbZ$.
\end{fact}

A \textit{Clifford circuit} is a circuit that consists of computational basis states being acted on by the basic Clifford gates, before being measured in the computational basis. Without loss of generality, we may assume that the input to the Clifford circuit is the all-zero state $\ket 0^{\otimes n}$. We define conjugated Clifford circuits (CCCs) similarly to Clifford circuits, except that each basic Clifford gate $G$ is replaced by a conjugated basic Clifford gate $(U^{\otimes k})^ \dag g U^{\otimes k}$, where $k =1$ when $g = H, S$ and $k=2$ when $g=\mathrm{CNOT}$. In other words,
\begin{defn} Let $U$ be a single-qubit unitary gate. A $U$-conjugated Clifford circuit ($U$-CCC) on $n$ qubits is defined to be a quantum circuit with the following structure:
\begin{enumerate}
\item Start with $\ket 0^{\otimes n}$.
\item Apply gates from the set $\{U^\dag HU, U^\dag SU,  (U^\dag \otimes U^\dag)\mathrm{CNOT}(U \otimes U) \}$.
\item Measure each qubit in the computational basis.
\end{enumerate}
\end{defn}
\noindent Because the intermediate $U$ and $U^\dag$ gates cancel, we may equivalently describe a $U$-CCC as follows:
\begin{enumerate}
\item Start with $\ket 0^{\otimes n}$.
\item Apply $U^{\otimes n}$.
\item Apply gates from the set $\{H, S, \mathrm{CNOT} \}$.
\item Apply $(U^\dag)^{\otimes n}$.
\item Measure each qubit in the computational basis.
\end{enumerate}

\subsection{Notions of classical simulation of quantum computation}

Let $\cP = \{p_z\}_z$ and $\cQ = \set{q_z}_z$ be (discrete) probability distributions, and let $\epsilon \geq 0$. We say that $\cQ$ is a \textit{multiplicative} $\epsilon$-\textit{approximation} of $\cP$ if for all $z$,
\begin{equation} \label{eq:multiplicativeapprox}
|p_z - q_z| \leq \epsilon p_z.
\end{equation}
We say that $\cQ$ is an \textit{additive} $\epsilon$-\textit{approximation} of $\cP$ if
\begin{equation} \label{eq:additiveapprox}
\frac{1}{2}\sum_z|p_z - q_z| \leq \epsilon.
\end{equation}

Note that any multiplicative $\epsilon$-approximation is also an additive $\epsilon/2$-approximation, since summing \eq{eq:multiplicativeapprox} over all $z$ produces \eq{eq:additiveapprox}. Here the factor of 1/2 is present so that $\epsilon$ is the total variation distance between the probability distributions.

A \textit{weak simulation with multiplicative (additive) error} $\epsilon >0$ of a family of quantum circuits is a classical randomized algorithm that samples from a distribution that is a multiplicative (additive) $\epsilon$-approximation of the output distribution of the circuit.
Note that from an experimental perspective, additive error is the more appropriate choice, since the fault-tolerance theorem merely guarantees additive closeness between the ideal and realized output distributions \cite{aharonov1997fault}.

There are of course other notions of simulability of quantum circuits -- such as strong simulation where one can compute individual output probabilities. We discuss these further in Section \ref{sec:simtypes}.

\subsection{Postselection gadgets}

Our results involve the use of postselection gadgets to simulate unitary operations. In this section, we introduce some terminology to describe these gadgets.

\begin{defn} \label{def:UCCCPostselectionGadget}
Let $U$ be a single-qubit operation. Let $k,l \in \bbZ^+$ with $k>l$. A $k$-\textit{to}-$l$ $U$-\textit{CCC postselection gadget} $G$ is a postselected circuit fragment that performs the following procedure on an $l$-qubit system:
\begin{enumerate}
\item Introduce a set $T$ of $(k-l)$ ancilla registers in the state $\ket{a_1\ldots a_{k-l}}$, where $a_1\ldots a_{k-l} \in \{0,1\}^{k-l}$.
\item Apply $U^{\otimes (k-l)}$ to the set $T$ of registers.
\item Apply a $k$-qubit Clifford operation $\Gamma$ to both the system and ancilla.
\item Choose a subset $S$ of $(k-l)$ registers and apply $(U^{\dag})^{\otimes (k-l)}$ to $S$.
\item Postselect on the subset $S$ of qubits being in the state $\ket{b_1 \ldots b_{k-l}}$, where $b_1\ldots b_{k-l} \in \{0,1\}^{k-l}$.
\end{enumerate}

An example of a $4$-to-$1$ $U$-CCC postselection gadget is the circuit fragment described by the following diagram:
\begin{eqnarray*}
\Qcircuit @C=1em @R=1em {
& \qw &  \qw &  \qw   &\multigate{3}{ \ \Gamma  \ }  & \gate{ U^\dag}   &  \rstick{\bra {b_1}}\qw \\ 
& &\lstick{\ket{a_1}}   & \gate U    &\ghost{ \ \Gamma \ }   & \gate{ U^\dag}   &  \rstick{\bra {b_2}}\qw  \\
&  & \lstick{\ket{a_2}}   & \gate U   &\ghost{ \ \Gamma \ } &      \gate{ U^\dag}   &  \rstick{\bra {b_3}}\qw  &    \\
& &\lstick{\ket{a_3}} & \gate U &  \ghost{ \ C \ }    &  \qw & \qw & \qw & \qw \\  \\
} 
\end{eqnarray*}

Let $G$ be a $U$-CCC postselection gadget as described in Definition \ref{def:UCCCPostselectionGadget}. The \textit{action} $A(G)$ (also denoted $A_G$) of $G$ is defined to be the linear operation that it performs, i.e.
\begin{equation}
A(G) = A_G = \bra{b_1\ldots b_l}_S \left( \prod_{i \in S} U_i^\dag \right) \Gamma \left( \prod_{i \in T} U_i \right) \ket{a_1\ldots a_l}_T ,
\end{equation}
and the \textit{normalized action} of $G$, when it exists, is 
\begin{equation} \label{eq:normalizedAction}
\tilde A_G = \frac{A_G}{(\det A_G)^{2^{-l}}}.
\end{equation}
\end{defn}

\noindent Note that the above normalization is chosen so that $\det \tilde A_G = 1$.

We say that a $U$-CCC postselection gadget $G$ is \textit{unitary} if there exists $\alpha \in \bbC\backslash \{0\}$ and a unitary operator $U$ such that $A_G = \alpha U$. It is straightforward to check that the following are equivalent conditions for gadget unitarity.
\begin{lemma} \label{lem:unitaryCCCgadget}
A $U$-CCC postselection gadget $G$ is unitary if and only if either one of the following holds:
\begin{enumerate}
\item There exists $\gamma >0$ such that $A_G^\dag A_G = \gamma I$,
\item $\tilde A_G^\dag \tilde A_G = I$, i.e. $\tilde A_G $ is unitary.
\end{enumerate}
\end{lemma}

Similarly, we say that a $U$-CCC postselection gadget $G$ is \textit{Clifford} if there exists $\alpha \in \bbC\backslash \{0\}$ and a Clifford operator $U$ such that $A_G = \alpha U$. 
The following lemma gives a necessary condition for a gadget to be Clifford.

\begin{lemma}  \label{lem:AGXAG}
If $G$ is a Clifford $U$-CCC postselection gadget, then
\begin{equation} \label{eq:AGXAG}
A_G X A_G^\dag \propto X \mbox{ or } A_G X A_G^\dag \propto Y \mbox{  or  } A_G X A_G^\dag \propto Z ,
\end{equation}
and
\begin{equation}  \label{eq:AGZAG}
A_G Z A_G^\dag \propto X \mbox{ or } A_G Z A_G^\dag \propto Y \mbox{  or  } A_G Z A_G^\dag \propto Z.
\end{equation}
\end{lemma}
\begin{proof}
If $G$ is a Clifford $U$-CCC postselection gadget, then there exists $\alpha \in \bbC\backslash\{0\}$ and a Clifford operation $\Gamma$ such that $A_G = \alpha \Gamma$. Since $\Gamma$ is Clifford, $\Gamma X\Gamma^\dag$ is a Pauli operator. But $\Gamma X\Gamma^\dag \not\sim I$, otherwise, $X \sim I$, which is a contradiction. Hence,  $\Gamma X\Gamma^\dag \sim X$ or $Y$ or $Z$, which implies \eq{eq:AGXAG}. The proof of  \eq{eq:AGZAG} is similar, with $X$ replaced with $Z$.
\end{proof}

\section{Weak simulation of CCCs with multiplicative error}

\subsection{Classification results}

In this section, we classify the hardness of weakly simulating $U$-CCCs as we vary $U$. As we shall see, it turns out that the classical simulation complexities of the $U$-CCCs associated with this notion of simulation are all of the following two types: the $U$-CCCs are either efficiently simulable, or are hard to simulate to constant multiplicative error unless the polynomial hierarchy collapses. To facilitate exposition, we will introduce the following terminology to describe these two cases: Let $\mathcal C$ be a class of quantum circuits. Following the terminology in \cite{koh2015further}, we say that $\mathcal C$ is in $\PWEAK$ if it is efficiently simulable in the weak sense by a classical computer. We say that $\mathcal C$ is $\PH$-\textit{supreme} (or that it exhibits $\PH$-\textit{supremacy}) if it satisfies the property that if $\mathcal C$ is efficiently simulable in the weak sense by a classical computer to constant multiplicative error, then the polynomial hierarchy ($\PH$) collapses.

The approach we take to classifying the $U$-CCCs is to decompose each $U$ into the form given by \eq{eq:Udecomposition}, 
\begin{equation}
U = e^{i \alpha} R_z(\phi) R_x(\theta) R_z(\lambda),
\end{equation}
and study how the classical simulation complexity changes as we vary $\alpha, \phi, \theta$ and $\lambda$. Two simplifications can immediately be made. First, the outcome probabilities of the $U$-CCC are independent of $\alpha$, since $\alpha$ appears only in a global phase. Second, the probabilities are also independent of $\lambda$. To see this, note that the outcome probabilities are all of the form: 
\begin{equation} \label{eq:probexpression}
|\bra b R_z(-\lambda)^{\otimes n} V R_z(\lambda)^{\otimes n} \ket 0|^2 = |\bra b V \ket 0|^2 ,
\end{equation}
which is independent of $\lambda$. In the above expression, $b \in \{0,1\}^n$ and $$V = R_x(-\theta)^{\otimes n} R_z(-\phi)^{\otimes n} \Gamma R_z(\phi)^{\otimes n} R_x(\theta)^{\otimes n}$$ for some Clifford circuit $\Gamma$. The equality follows from the fact that the computational basis states are eigenstates of $R_z(\lambda)^{\otimes n}$ with unit-magnitude eigenvalues.

Hence, to complete the classification, it suffices to just restrict our attention to the two-parameter family $\set{R_z(\phi) R_x(\theta)}_{\phi,\theta}$ of unitaries. We first prove the following lemma (see Table \ref{tab:classification} for a summary):
\begin{lemma} \label{thm:multHardness}
Let $U = R_z(\phi) R_x(\theta)$, where $\phi, \theta \in [0,2\pi)$. Then
\begin{itemize}
\item $U$-CCCs are in $\PWEAK$, if
\begin{enumerate}
\itemsep1em 
\item[(i)] $\phi \in [0,2\pi)$ \ and  \ $\theta \in \pi \bbZ$, or
\item[(ii)] $\phi \in \frac \pi 2 \bbZ$ \ and \ $\theta \in \frac\pi 2 \bbZ$.
\end{enumerate}
\item $U$-CCCs are $\PH$-supreme, if
\begin{enumerate}
\itemsep1em 
\item[(iii)] $\phi \notin \frac \pi 2 \bbZ$ \ and \  $\theta \in \frac \pi 2 \bbZ_{odd}$, or
\item[(iv)] $\theta \notin \frac \pi 2 \bbZ$.
\end{enumerate}
\end{itemize}
\end{lemma}

\begin{table}
\begin{center}
\begin{tabular}{|C{1.2cm}||C{2cm}|C{2cm}|C{2.1cm}|}
\hline
\backslashbox[16.1mm]{\raisebox{0.4mm}{$\phi$}}{\raisebox{-1.4mm}{$\theta$}} & $\pi \bbZ$ & $\tfrac \pi 2 \bbZ_{odd}$ & $\left(\tfrac \pi 2 \bbZ\right)^c$ \\ \hline \hline
$\tfrac \pi 2 \bbZ$ & \small $\PWEAK$ \newline \footnotesize (i, ii) & \small $\PWEAK$  \newline \footnotesize (ii) & \small $\PH$-supreme \newline \footnotesize (iv) \\ \hline
$\left(\tfrac \pi 2 \bbZ\right)^c$ & \small $\PWEAK$ \newline \footnotesize (i) & \small  $\PH$-supreme \newline \footnotesize (iii) & \small $\PH$-supreme \newline \footnotesize (iv) \\ \hline
\end{tabular}
\end{center}
\caption{\small Complete complexity classification of $U$-CCCs (where $U= R_z(\phi) R_x(\theta)$) with respect to weak simulation, as we vary $\phi$ and $\theta$. The roman numerals in parentheses indicate the parts of Lemma \ref{thm:multHardness} that are relevant to the corresponding box. All $U$-CCCs are either in $\PWEAK$ (i.e. can be efficiently simulated in the weak sense) or $\PH$-supreme (i.e. cannot be simulated efficiently in the weak sense, unless the polynomial hierarchy collapses.)
\newline\newline }
\label{tab:classification}
\end{table}

We defer the proof of Lemma \ref{thm:multHardness} to Sections \ref{sec:efficientClassicalSim} and \ref{sec:ProofOfMultHardness}. Lemma \ref{thm:multHardness} allows us to prove our main theorem:
\begin{theo} \label{cor:classificationWeak}
Let $U$ be a single-qubit unitary operator. Consider the following two statements:
\begin{enumerate}
\item[(A)] $U$-CCC is in $\PWEAK$.
\item[(B)] There exists a single-qubit Clifford operator $\Gamma \in \langle S, H \rangle$ and $\lambda \in [0,2\pi)$ such that\footnote{or alternatively, we could restrict the range of $\lambda$ to be in $[0,\pi]$, since any factor of $R_z(\pi/2) \sim S$ can be absorbed into the Clifford operator $\Gamma$.} 
\begin{equation} \label{eq:CRzlambda}
U \sim \Gamma R_z(\lambda) .
\end{equation}
\end{enumerate}
Then, 
\begin{enumerate}
\item (B) implies (A).
\item If the polynomial hierarchy is infinite, then (A) implies (B).
\end{enumerate}

In other words, if we assume that the polynomial hierarchy is infinite, then $U$-CCCs are $\PH$-supreme if and only if they cannot be written in the form $U \sim \Gamma R_z(\lambda)$, where $\Gamma$ is a Clifford circuit and $R_z(\lambda)$ is a $Z$-rotation.
\end{theo}

\begin{proof} \hfill
\begin{enumerate}
\item Since $R_z(\lambda) \ket 0 \sim \ket 0$, it follows that for any $\Gamma$, $\Gamma R_z(\lambda)$-CCCs have the same outcome probabilities as $\Gamma$-CCCs. But $C$-CCCs are efficiently simulable, by the Gottesman-Knill Theorem, since $\Gamma \in \langle S, H \rangle$. Hence, $U$-CCCs are in $\PWEAK$.
\item Let $U$ be such that $U$-CCCs are in $\PWEAK$. Using the decomposition in \eq{eq:Udecomposition}, write $U = e^{i \alpha} R_z(\phi) R_x(\theta) R_z(\lambda)$. Since we assumed that the polynomial hierarchy is infinite, Lemma \ref{thm:multHardness} implies that 
\begin{enumerate}
\item $\theta \in \pi \bbZ$, or
\item $\theta \in \tfrac \pi 2 \bbZ$ and  $\phi \in \tfrac \pi 2 \bbZ$.
\end{enumerate} 

In Case (a), $\theta \in 2\pi \bbZ$ or $\pi \bbZ_{odd}$. If $\theta \in 2\pi \bbZ$, then $$U \sim R_z(\phi) R_x(2\pi \bbZ) R_z(\gamma) = I. R_z(\phi+\gamma),$$ which is of the form given by \eq{eq:CRzlambda}. If $\pi \bbZ_{odd}$, then $$U \sim R_z(\phi) R_x(\pi \bbZ_{odd}) R_z(\gamma) \sim R_z(\phi) X R_z(\gamma) = X R_z(\gamma - \phi),$$ which is again of the form given by \eq{eq:CRzlambda}. 

In Case (b), 
\begin{eqnarray}
U &\in & e^{i\alpha} R_z(\pi \bbZ/2 ) R_x(\pi \bbZ/2 ) R_z(\gamma) \nonumber\\
&=& e^{i\alpha} R_z(\pi \bbZ/2 ) H R_z(\pi \bbZ/2 ) H R_z(\gamma) .
\end{eqnarray}
But the elements of $R_z(\pi \bbZ/2 )$ are of the form $S^j$, for $j \in \bbZ$, up to a global phase. Therefore, $R_z(\pi \bbZ/2 ) H R_z(\pi \bbZ/2 ) H $ is Clifford, and $U$ is of the form \eq{eq:CRzlambda}. 
\end{enumerate}
\end{proof}

Hence, Theorem \ref{cor:classificationWeak} tells us that under the assumption that the polynomial hierarchy is infinite, $U$-CCCs can be simulated efficiently (in the weak sense) if and only if $U \sim \Gamma R_z(\lambda)$ for some single qubit Clifford operator $\Gamma$, i.e. if $U$ is a Clifford operation times a $Z$-rotation.

\subsection{Proofs of efficient classical simulation} \label{sec:efficientClassicalSim}

In this section, we prove Cases (i) and (ii) of Lemma \ref{thm:multHardness}.

\subsubsection{Proof of Case (i): $\phi \in [0,2\pi)$ \ and  \ $\theta \in \pi \bbZ$} \label{sec:Case1}

\begin{theo}
Let $U = R_z(\phi) R_x(\theta)$. If $\phi \in [0,2\pi)$ and $\theta \in \pi \bbZ$, then $U$-CCCs are in $\PWEAK$.
\end{theo}
\begin{proof}
First, we consider the case where $\theta \in 2\pi \bbZ$. In this case, $U = R_z(\phi)$, and the amplitudes of the $U$-CCC can be written as
\begin{equation}
\bra y R_z(-\phi)^{\otimes n} \Gamma R_z(\phi)^{\otimes n} \ket x \sim \bra y \Gamma \ket x
\end{equation}
for some Clifford operation $\Gamma $ and computational basis states $\ket x$ and $\ket y$.
By the Gottesman-Knill Theorem, these $U$-CCCs can be efficiently weakly simulated.

Next, we consider the case where $\theta \in \pi \bbZ_{odd}$. In this case, $U = R_z(\phi)R_x(\pi) \sim R_z(\phi) X$, and the amplitudes of the $U$-CCC can be written as
\begin{equation}
\bra y X^{\otimes n} R_z(-\phi)^{\otimes n} \Gamma R_z(\phi)^{\otimes n} X^{\otimes n} \ket x \sim \bra {\bar y} \Gamma \ket {\bar x}
\end{equation}
for some Clifford operation $\Gamma$ and computational basis states $\ket x$ and $\ket y$, where $\bar z$ is the bitwise negation of $z$.
By the Gottesman-Knill Theorem, these $U$-CCCs can be efficiently weakly simulated.

Putting the above results together, we get that $U$-CCCs are in $\PWEAK$.
\end{proof}

\subsubsection{Proof of Case (ii): $\phi \in \tfrac \pi 2 \bbZ$ \ and \ $\theta \in \tfrac\pi 2 \bbZ$} \label{sec:Case2}
\begin{theo}
Let $U = R_z(\phi) R_x(\theta)$. If $\phi \in \tfrac \pi 2 \bbZ$ and $\theta \in \tfrac\pi 2 \bbZ$, then $U$-CCCs are in $\PWEAK$.
\end{theo}

\begin{proof}

The elements of $R_z(\tfrac \pi 2 \bbZ )$ are of the form $S^j$, where $j \in \bbZ$, up to a global phase. Therefore, $U = R_z(\phi) R_x(\theta) = R_z(\phi) H R_z(\theta) H $ is a Clifford operation, and so, the $U$-CCCs consist of only Clifford gates. By the Gottesman-Knill Theorem, these $U$-CCCs can be be efficiently (weakly) simulated.
\end{proof}

\subsection{Proofs of hardness} \label{sec:ProofOfMultHardness}

In this section, we prove Cases (iii) and (iv) of Lemma \ref{thm:multHardness}. Our proof uses postselection gadgets, similar to the techniques used in \cite{BJS2010,boulandccc2016}. One can also prove hardness using techniques from measurement-based-quantum computing, at least for certain $U$. We give such a proof in Appendix \ref{app:oldproof} for the interested reader; we believe this proof may be more intuitive for those who are familiar with measurement-based quantum computing.

We start by proving a lemma that will be useful for the proofs of hardness.

\begin{lemma} \label{lem:suffPHsup}
(Sufficient condition for $\PH$-supremacy) Let $U$ be a single-qubit gate. If there exists a unitary non-Clifford $U$-CCC postselection gadget $G$, then $U$-CCCs are $\PH$-supreme.
\end{lemma}
\begin{proof}

Suppose such a gadget $G$ exists. Then, since the Clifford group plus any non-Clifford gate is universal \cite{nebe2001invariants,nebe2006self,cliffordstackexchange}, the Clifford group plus $G$ must be universal on a single qubit.
 Then, by the inverse-free Solovay-Kitaev Theorem of Sardharwalla \emph{et al.} \cite{sardharwalla2016universal}, using polynomially many gates from the set $G,H,S$ one can compile any desired one-qubit unitary $V$ to inverse exponential accuracy (since in particular $\langle H,S\rangle$ contains the Paulis). In particular, since any three-qubit unitary can be expressed as a product of a constant number of  CNOTs and one-qubit unitaries, one can compile any gate in the set $\{$CCZ, Controlled-H, all one-qubit gates $\}$ to inverse exponential accuracy with polynomial overheard.

In his proof that $\PostBQP=\PP$, Aaronson showed that postselected poly-sized circuits of the above gates can compute any language in $\PP$  \cite{aaronson2005quantum}. Furthermore, as his postselection succeeds with inverse exponential probability, compiling these gates to inverse exponential accuracy is sufficient for performing arbitrary $\PP$ computations. 

Hence, by using polynomially many gadgets for $G$, CNOT, $H$ and $S$, one can compile Aaronson's circuits\footnote{More specifically, we compile the circuit given by $(U^\dagger)^{\otimes n}$, then Aaronson's circuit, then $U^{\otimes n}$, as we need to cancel the $U$'s at the beginning and the $U^\dagger$s at the end in order to perform Aaronson's circuit which starts and measures in the computational basis. However as the $U,U^\dagger$ are one-qubit gates, one can cancel them to inverse exponential accuracy using our gates, and hence this construction suffices.} for computing $\PP$ to inverse exponential accuracy, and hence these circuits can compute $\PP$-hard problems. $\PH$-supremacy then follows from the techniques of \cite{BJS2010, aaronsonboson}. Namely, a weak simulation of such circuits with constant multiplicative error would place $\PP\subseteq \BPP^\NP \subseteq \Delta_3$ by Stockmeyer counting, and hence by Toda's theorem this would result in the collapse of $\PH$ to the third level.  In fact, by the arguments of Fujii \emph{et al.} \cite{fujii2014impossibility}, one can collapse $\PH$ to the second level as well, by placing $\mathsf{coC}_=\mathsf{P}$ in $\mathsf{SBP}$, and we refer the interested reader to their work for the complete argument. 

\end{proof}

\subsubsection{Proof of Case (iii): $\phi \notin \tfrac \pi 2 \bbZ$ \ and \  $\theta \in \tfrac \pi 2 \bbZ_{odd}$} \label{sec:Case3}

Let $U= R_z(\phi) R_x(\theta)$. Consider the following $U$-CCC postselection gadget:
\begin{equation}\label{eq:gadgetI} 
I(\phi,\theta) = \qquad 
\Qcircuit @C=1em @R=1em {
& \qw & \qw &  \ctrl{1}  & \gate{U^\dag}  & \qw & \bra{0}   \\
& \lstick{\ket{0}} & \gate{U} &  \control \qw  & \qw & \qw & \qw
}
\end{equation}

We now prove some properties about $I(\phi, \theta)$.

\begin{theo} \label{thm:propI} \hfill
\begin{enumerate} 
\item  The action of $I(\phi,\theta)$ is
\begin{equation} \label{eq:actionI}
A_{I(\phi,\theta)} = \mat{\cos^2 \tfrac \theta 2 & \tfrac i2 \sin\theta \ e^{-i \phi} \\ -\tfrac i2 \sin\theta \ e^{i \phi} & -\sin^2 \tfrac \theta 2}.
\end{equation}
\item $I(\phi,\theta)$ is a unitary gadget if and only if $\theta \in \tfrac \pi 2 \bbZ_{odd}$. When $I(\phi,\theta)$ is unitary,
\begin{equation} \label{eq:AtidleI}
\tilde A_{I(\phi,\theta)} = \frac i{\sqrt 2} \mat{1 & i (-1)^k e^{-i\phi} \\ -i(-1)^k e^{i\phi} & -1} ,
\end{equation}
where $k = \tfrac \theta \pi - \tfrac 12$.
\item $I(\phi,\theta)$ is a Clifford gadget if and only if $\phi \in \tfrac \pi 2 \bbZ$ and $\theta \in \tfrac \pi 2 \bbZ_{odd}$.
\item $I(\phi,\theta)$ is a unitary non-Clifford gadget if and only if $\phi \notin \tfrac \pi 2 \bbZ$ and $\theta \in \tfrac \pi 2 \bbZ_{odd}$.
\end{enumerate}
\end{theo}
\begin{proof} \hfill
\begin{enumerate}
\item By direct calculation.
\item By \eq{eq:actionI}, 
\begin{equation}
A_{I(\phi,\theta)}^\dag A_{I(\phi,\theta)} = \mat{\cos^2 \tfrac \theta 2 & \tfrac i4 \sin(2\theta) e^{-i \phi} \\ -\tfrac i4 \sin(2\theta) e^{i \phi} & \sin^2 \tfrac \theta 2} .
\end{equation}
If $\theta \in \tfrac \pi 2 \bbZ_{odd}$, then $A_{I(\phi,\theta)}^\dag A_{I(\phi,\theta)} = \tfrac 12 I$, which implies that $I(\phi,\theta)$ is a unitary gadget, by Lemma \ref{lem:unitaryCCCgadget}.
Conversely, assume that $I(\phi,\theta)$ is a unitary gadget. Suppose that $\theta \notin \tfrac \pi 2 \bbZ_{odd}$. Then $\sin(2\theta) \neq 0$, which implies that $A_{I(\phi,\theta)}^\dag A_{I(\phi,\theta)}  \not\propto I$, which is a contradiction. Hence, $\theta \in \tfrac \pi 2 \bbZ_{odd}$.

Next, $k = \tfrac \theta \pi - \tfrac 12$ implies that $\theta = \tfrac\pi 2(2k+1)$. Since $\theta \in \tfrac \pi 2 \bbZ_{odd}$, it follows that $k \in \bbZ$. Then $\sin \theta = (-1)^k$, $\cos^2 \tfrac \theta 2 = \tfrac 12$ and $\sin^2 \tfrac \theta 2 = \tfrac 12$. 
Hence,
\begin{equation} \label{eq:AIphitheta}
A_{I(\phi,\theta)} = \mat{\tfrac 12 & \tfrac i2 (-1)^k e^{-i\phi} \\ -\tfrac i2 (-1)^k e^{i \phi} & -\tfrac 12} .
\end{equation}
Hence, $\det A_{I(\phi,\theta)} = -\tfrac 12$.  Plugging this and \eq{eq:AIphitheta} into \eq{eq:normalizedAction} gives \eq{eq:AtidleI}.
\item $(\Leftarrow)$ Let $\phi \in \tfrac \pi 2 \bbZ$ and $\theta \in \tfrac \pi 2 \bbZ_{odd}$. Write $\phi = \tfrac \pi 2 l$ and $\theta = \tfrac \pi 2(2k+1)$. Then, by \eq{eq:AtidleI},
\begin{equation}
\tilde A_{I(\phi,\theta)} = \frac i{\sqrt 2} \mat{1 & i^{1+2k+3l} \\ i^{3+2k+l} & -1}.
\end{equation}
Now, it is straightforward to check that for all $k,l \in \bbZ$,  $\tilde A_{I(\phi,\theta)} X \tilde A_{I(\phi,\theta)}^\dag \in \{-X,Z,-Z\}$ and $\tilde A_{I(\phi,\theta)} Z \tilde A_{I(\phi,\theta)}^\dag \in \{-Y,X,Y,-X\}$. This shows that $\tilde A_{I(\phi,\theta)}$ maps the Pauli group to itself, under conjugation, which implies that $\tilde A_{I(\phi,\theta)}$ is Clifford.

$(\Rightarrow)$ Assume that $I(\phi,\theta)$ is a Clifford gadget. Suppose that $\phi \notin \tfrac \pi 2 \bbZ$ or $\theta \notin \tfrac \pi 2 \bbZ_{odd}$. But $I(\phi,\theta)$ is unitary, and hence, $\theta \in \tfrac \pi 2 \bbZ_{odd}$. So $\phi \notin \tfrac \pi 2 \bbZ$. By Lemma \ref{lem:AGXAG}, $\tilde A_{I(\phi,\theta)} X \tilde A_{I(\phi,\theta)}^\dag \sim X$ or $Y$ or $Z$. But, as we compute,
\begin{equation}
\tilde A_{I(\phi,\theta)} X \tilde A_{I(\phi,\theta)}^\dag = \mat{(-1)^k \sin \phi & -e^{-i\phi} \cos \phi \\ -e^{i\phi} \cos\phi & -(-1)^k \sin \phi} .
\end{equation}

If $\tilde A_{I(\phi,\theta)} X \tilde A_{I(\phi,\theta)}^\dag \sim X$ or $Y$, then $\sin \phi =0$, which is a contradiction, since $\phi \notin \tfrac \pi 2 \bbZ$. Hence, $\tilde A_{I(\phi,\theta)} X \tilde A_{I(\phi,\theta)}^\dag \sim Z$, which implies that $\cos\phi = 0$. But this also contradicts $\phi \notin \tfrac \pi 2 \bbZ$. Hence, $\phi \in \tfrac \pi 2 \bbZ$ and $\theta \in \tfrac \pi 2 \bbZ_{odd}$.

\item Follows from Parts 2 and 3 of Theorem \ref{thm:propI}.
\end{enumerate}
\end{proof}

\begin{theo}
Let $U = R_z(\phi) R_x(\theta)$. If $\phi \notin \tfrac \pi 2 \bbZ$ and $\theta \in \tfrac \pi 2 \bbZ_{odd}$, then $U$-CCCs are $\PH$-supreme.
\end{theo}

\begin{proof}
By Theorem \ref{thm:propI}, when $\phi \notin \tfrac \pi 2 \bbZ$ and $\theta \in \tfrac \pi 2 \bbZ_{odd}$, then $I(\phi,\theta)$ is a unitary non-Clifford $U$-CCC postselection gadget. Hence, by Lemma \ref{lem:suffPHsup}, $U$-CCCs are $\PH$-supreme.
\end{proof}

\subsubsection{Proof of Case (iv): $\theta \notin \tfrac \pi 2 \bbZ$} \label{sec:Case4}

Let $U= R_z(\phi) R_x(\theta)$. Consider the following $U$-CCC postselection gadget:
\begin{equation}\label{eq:gadgetJ} 
J(\phi,\theta) = \qquad 
\Qcircuit @C=1em @R=1em {
& \qw & \qw & \qw & \ctrl{1}  & \qw  & \qw & \qw   \\
& \lstick{\ket{0}} & \gate{U} & \gate S &  \control \qw  & \gate{U^\dag}  & \qw & \bra{0}
}
\end{equation}

We now prove some properties about $J(\phi, \theta)$.

\begin{theo} \label{thm:propJ} \hfill
\begin{enumerate} 
\item  The action of $J(\phi,\theta)$ is
\begin{eqnarray} \label{eq:actionJ}
A_{J(\phi,\theta)} &=&
\frac 1{\sqrt 2} e^{-i \frac \pi 4} \mat{i + \cos \theta & 0 \\ 0 & 1+ i \cos\theta} \nonumber\\
&=& \frac i{\sqrt 2} e^{-i \tfrac \pi 4} \sqrt{1+\cos^2 \theta} \ S^\dag R_z(2 \tan^{-1}( \cos\theta )).
\end{eqnarray}
\item $J(\phi,\theta)$ is a unitary gadget for all $\theta, \phi \in [0,2\pi)$. The normalized action is 
\begin{equation}
\tilde A_{J(\phi,\theta)} \sim S^\dag R_z(2 \tan^{-1}( \cos\theta )).
\end{equation} 
\item $J(\phi,\theta)$ is a Clifford gadget if and only if $\theta \in \tfrac \pi 2 \bbZ$.
\item $J(\phi,\theta)$ is a unitary non-Clifford gadget if and only if $\theta \notin \tfrac \pi 2 \bbZ$.
\end{enumerate}
\end{theo}

\begin{proof} \hfill
\begin{enumerate}
\item By direct calculation.
\item The determinant of $A_{J(\phi,\theta)}$ is
\begin{equation}
\det A_{J(\phi,\theta)} = \tfrac 12 (1+ \cos^2 \theta) \neq 0
\end{equation}
for all $\theta$ and $\phi$. Hence, $A_{J(\phi, \theta)} \propto S^\dag R_z(2 \tan^{-1}( \cos\theta ))$ for all $\theta$ and $\phi$, which implies that $J(\phi, \theta)$ is a unitary gadget for all $\theta$ and $\phi$. 

Hence, $$\tilde A_{J(\phi,\theta)} = \frac{A_{J(\phi,\theta)}}{\sqrt{\det A_{J(\phi,\theta)}}} = i e^{-i \frac \pi 4} S^\dag R_z(2 \tan^{-1}( \cos\theta )). $$
\item 
\begin{eqnarray} 
J(\phi,\theta) \mbox{ is a Clifford gadget} 
& \Leftrightarrow & S^\dag R_z(2 \tan^{-1}( \cos\theta )) \mbox{ is Clifford} \nonumber\\
& \Leftrightarrow & R_z(2 \tan^{-1}( \cos\theta )) \mbox{ is Clifford} \nonumber\\
& \Leftrightarrow & 2 \tan^{-1}( \cos\theta ) \in \tfrac \pi 2 \bbZ  \quad \mbox{ by Fact \ref{fact:RzClifford}} \nonumber\\
& \Leftrightarrow & \cos \theta \in \{0,1,-1\} \nonumber\\
& \Leftrightarrow & \theta \in \tfrac \pi 2 \bbZ .
\end{eqnarray}
\item  Follows from Parts 2 and 3 of Theorem \ref{thm:propJ}.
\end{enumerate}

\end{proof}

\begin{theo}
Let $U = R_z(\phi) R_x(\theta)$. If $\theta \notin \tfrac \pi 2 \bbZ$, then $U$-CCCs are $\PH$-supreme.
\end{theo}

\begin{proof}
By Theorem \ref{thm:propJ}, when $\theta \notin \tfrac \pi 2 \bbZ$, then $I(\phi,\theta)$ is a unitary non-Clifford $U$-CCC postselection gadget. Hence, by Lemma \ref{lem:suffPHsup}, $U$-CCCs are $\PH$-supreme.
\end{proof}

\section{Weak simulation of CCCs with additive error}
\label{sec:additive}

Here we show how to achieve additive hardness of simulating conjugated Clifford circuits, under additional hardness assumptions. Specifically, we will show that under these assumptions, there is no classical randomized algorithm which given a one-qubit unitary $U$ and a Clifford circuit $V$, samples the output distribution of $V$ conjugated by $U$'s up to constant $\ell_1$ error.

In the following, let $V$ be a Clifford circuit on $n$ qubits, $U$ be a one-qubit unitary which is not a $Z$-rotation times a Clifford, and $y\in\{0,1\}^n$ be an $n$-bit string. Define  
\[
p_{y,U,V}=\left|\bra{y} (U^\dagger)^{\otimes n} V U^{\otimes n} \ket{0^n}\right|^2.
\]
In other words $p_{y,U,V}$ is the probability of outputting the string $y$ when applying the circuit $V$ conjugated by $U$'s to the all $0$'s state, and then measuring in the computational basis. Let the corresponding probability distribution on $y$'s given $U$ and $V$ be denoted $D(U,V)$.

\begin{theo} Assuming that PH is infinite and Conjecture \ref{conj:specificavghardness}, then there is no classical algorithm which given a one-qubit unitary $U$ and an $n$-qubit Clifford circuit $V$, outputs a probability distribution which is $1/100$ close to $D(U,V)$ in total variation distance.
\end{theo}

\begin{conj}
For any $U$ which is not equal to a Z-rotation times a Clifford, it is $\#\mathsf{P}$-hard to approximate a $6/50$ fraction of the $p_{y,U,V}$ over the choice of $y,V$ to within multiplicative error $1/2+o(1)$.
\label{conj:specificavghardness}
\end{conj}

In order to prove this we'll actually prove a more general theorem described below; the result will then follow from simply setting $a=c=1/5$, $\epsilon=1/100$. One can in general plug in any values they like subject to the constraints; for instance one can strengthen the hardness assumption by assuming computing a smaller fraction of the $p_{y,U,V}$ is still $\#\mathsf{P}$-hard to obtain larger allowable error in the simulation. These parameters are similar to those appearing in other hardness conjectures, for example those used for $\IQP$ \cite{bremner2016average}.

\begin{theo}
\label{l1hardness}
Pick constants $0<\epsilon,a,c <1$ such that $(1-a)^2/2-c>0$ and $\frac{2\epsilon}{ac}<1$. Then assuming Conjecture \ref{conj:avghardness}, given a one-qubit unitary $U$ and an $n$-qubit Clifford circuit $V$, one cannot weakly simulate the distribution $D(U,V)$ with a randomized classical algorithm with total variation distance error $\epsilon$, unless the polynomial hierarchy collapses to the third level. 
\end{theo}

\begin{conj} For any $U$ which is not equal to a Z-rotation times a Clifford, it is $\#\mathsf{P}$-hard to multiplicatively approximate $(1-a)^2/2-c$ fraction of the $p_{y,U,V}$ over the choice of $(y,V)$, up to multiplicative error $\frac{2\epsilon}{ac}+o(1)$.
\label{conj:avghardness}
\end{conj}

\begin{proof}[Proof of Theorem \ref{l1hardness}]

Suppose by way of contradiction that there exists a classical poly-time randomized algorithm which given inputs $U,V$ outputs samples from a distribution $D'(U,V)$ such that $\frac{1}{2}|D(U,V)-D'(U,V)|_1 < \epsilon$.
In particular, let $q_{y,U,V}$ be the probability that $D'(U,V)$ outputs $y$ -- i.e. the probability that the simulation outputs $y$ under inputs $U,V$.

By our simulation assumption, for all $U,V$ we have that
$\sum_{y} |q_{y,U,V} - p_{y,U,V}| \leq 2\epsilon$.
Therefore by Markov's inequality, given our constant $0<c<1$, we have that for all $U$ and $V$ there exists a set $S'\subseteq \{0,1\}^n$ of output strings $y$ of size $|S'|/2^n > 1-c$, such that for all $y\in S'$,
\[|q_{y,U,V} - p_{y,U,V}|\leq \frac{2\epsilon}{c 2^n}.\]

In particular, by averaging over $V$'s, we see that for any $U$ as above, there exists a set $S \subset \{0,1\}^n \times \mathcal{C}$ of pairs $(y,V)$ such that for all $(y,V)\in S$, $|q_{y,U,V} - p_{y,U,V}|\leq \frac{2\epsilon}{c 2^n}.$ Furthermore $S$ has measure at least $(1-c)$ over a uniformly random choice of $(y,V)$.

We now show the following anticoncentration lemma (similar theorems were shown independently in \cite{hangleiter2017anti, bremner2017upcoming,saeed2018upcoming}):
\begin{lemma} 
\label{anticonc}
For any fixed $U$ and $y$ as above, and for any constant $0<a<1$, we have that at least $\frac{(1-a)^2}{2}$ fraction of the Clifford circuits $V$ have the property that

\[p_{y,U,V} \geq \frac{a}{2^n} . \]

\end{lemma}

We will prove Lemma \ref{anticonc} shortly. First, we will show why this implies Theorem \ref{l1hardness}. 
In particular, by averaging Lemma \ref{anticonc} over $y$'s, we see that for any $U$ as above, there exists a set $T \subset \{0,1\}^n \times \mathcal{C}$ of pairs $(y,V)$ such that  for all $(y,V)\in T$, $p_{y,U,V} \geq \frac{a}{2^n}$. Furthermore $T$ has measure at least $\frac{(1-a)^2}{2}$ over a uniformly random choice of $(y,V)$.
Since we assumed that $(1-a)^2/2 + (1-c) >1$, then $S\cap T$ must be nonempty, and in particular must contain $(1-a)^2/2-c$ fraction of the pairs $(y,V)$. On this set $S\cap T$, we have that

\[q_{y,U,V} \leq p_{y,U,V} + \frac{2\epsilon}{c2^n} =p_{y,U,V} + \frac{2\epsilon}{ac}\frac{a}{2^n}  \leq  \left(1 + \frac{2\epsilon}{ac}\right)p_{y,U,V}, \]
and likewise
\[q_{y,U,V} \geq p_{y,U,V} - \frac{2\epsilon}{c2^n} =p_{y,U,V} - \frac{2\epsilon}{ac}\frac{a}{2^n}  \geq  \left(1 - \frac{2\epsilon}{ac}\right)p_{y,U,V} . \]

Since $1-\frac{2\epsilon}{ac}>0$ (which we guaranteed by assumption),  $q_{y,U,V}$ is a multiplicative approximation to $p_{y,U,V}$ with multiplicative error $\frac{2\eps}{ac}$ for $(y,V)$ in the set $S\cap T$. The set $S\cap T$ contains at least $(1-a)^2/2-c$ fraction of the total pairs $(y,V)$.

On the other hand, by Conjecture \ref{conj:avghardness} we have that computing a $ (1-a)^2/2 -c$ fraction of the $p_{y,U,V}$ to this level of multiplicative error is a  $\#\mathsf{P}$-hard task. So approximating $p_{y,U,V}$ to this level of multiplicative error for this fraction of outputs is both $\#\mathsf{P}$-hard, and achievable by our simulation algorithm. 
This collapses $\PH$ to the third level by known arguments \cite{aaronsonboson,BJS2010}. 
In particular, by applying Stockmeyer's approximate counting algorithm \cite{stockmeyer} to $p_{y,U,V}$, one can multiplicatively approximate $q_{y,U,V}$ to multiplicative error $\frac{1}{\mathsf{poly}}$ in $\mathsf{F}\BPP^\NP$ for those elements in $S\cap T$. 
But since $q_{y,U,V}$ is a  $\frac{2\epsilon}{ac}$-approx to $p_{y,U,V}$, this is a $\frac{2\epsilon}{ac}+o(1)$ multiplicative approximation to $p_{y,U,V}$ in $S\cap T$.
Hence a $\#\mathsf{P}$-hard quantity is in $\mathsf{F}\BPP^\NP$. 
This collapses $\PH$ to the third level by Toda's theorem \cite{Toda91}.

To complete our proof of Theorem \ref{l1hardness}, we will prove Lemma \ref{anticonc}.

\begin{proof}[Proof of Lemma \ref{anticonc}]

To prove this, we will make use of the fact that the Clifford group is an exact 2-design\footnote{The Clifford group is also a 3-design, but we will only need the fact it is a 2-design for our proof.} \cite{webbclifford,zhuclifford}. The fact that the Clifford group is a 2-design means that for any polynomial $p$ over the variables $\{V_{ij}\}$ and their complex conjugates, which is of degree at most 2 in the $V_{ij}$'s and degree at most 2 in the $V_{ij}^*$'s, we have that
\[\frac{1}{|C|} \sum_{V \in C} p(V,V^*) = \int p(V,V^*) \d V, \]
where $C$ denotes the Clifford group and the integral $\d V$ is taken over the Haar measure. In other words, the expectation values of low-degree polynomials in the entries of the matrices are exactly identical to the expectation values over the Haar measure. 

In particular, note that $p_{y,U,V}$ is a degree-1 polynomial in the entries of $V$ and their complex conjugates, and $p_{y,U,V}^2$ is a degree-2 polynomial in these variables. Therefore, since the Clifford group is an exact 2-design, we have that for any $y$ and $U$,
\[\frac{1}{|C|} \sum_{V \in \mathcal{C}} p_{y,U,V} = \int p_{y,U,V} \d V  = \frac{1}{2^n}\] 
and
\[\frac{1}{|C|} \sum_{V \in \mathcal{C}} p_{y,U,V}^2 = \int p_{y,U,V}^2 \d V = \frac{2}{2^{2n} - 1} \left(1 - \frac{1}{2^n}\right),\] 
where the values of these integrals over the Haar measure are well known -- see for instance Appendix D of \cite{harlow2016jerusalem}.

Following \cite{bremner2016average}, we now invoke the Paley-Zygmund inequality, which states that:

\begin{fact} Given a parameter $0<a<1$, and a non-negative random variable $p$ of finite variance, we have
\[\mathrm{Pr}[p\geq a \mathbb{E}[p]] \geq (1-a)^2 \mathbb{E}[p]^2/\mathbb{E}[p^2] .
\]
\end{fact}

Applying this inequality to the random variable $p_{y,U,V}$ over the choice of the Clifford circuit $V$, we have that
\[
\mathrm{Pr}_V\left[p_{y,U,V}\geq \frac{a}{2^n}\right] \geq (1-a)^2 \frac{2^{-2n}}{\frac{2 - 2^{-n+1}}{2^{2n}-1} } = (1-a)^2 \frac{1 - 2^{-2n}}{2 - 2^{-n+1}} \geq \frac{(1-a)^2}{2}
\]
which implies the claim.

\end{proof}

This completes the proof of Theorem \ref{l1hardness}.

\end{proof}

\section{Evidence in favor of hardness conjecture}

In Section \ref{sec:additive}, we saw that by assuming an average case hardness conjecture (namely Conjecture \ref{conj:avghardness}), we could show that a weak simulation of CCCs to additive error would collapse the polynomial hierarchy. A natural question is: what evidence do we have that Conjecture \ref{conj:avghardness} is true? 

In this section, we show that the worst-case version of Conjecture \ref{conj:avghardness} is true. In fact, we show that for any $U\neq C R_Z(\theta)$ for a Clifford $C$, there exists a Clifford circuit V and an output y such that computing $p_{y,U,V}$ is $\#\mathsf{P}$-hard to constant multiplicative error. Therefore certainly \emph{some} output probabilities of CCCs are $\#\mathsf{P}$-hard to compute. Conjecture \ref{conj:avghardness} is merely conjecturing further that computing a large fraction of such output probabilities is just as hard.

\begin{theo}[Worst-case version of Conjecture \ref{conj:avghardness}]
For any $U$ which is not equal to a Z-rotation times a Clifford, there exists a Clifford circuit $V$ and string $y\in\{0,1\}^n$ such that it is $\#\mathsf{P}$-hard to multiplicatively approximate a $p_{y,U,V}$ to multiplicative error $1/2 -o(1)$.
\label{thm:worstcasehardness}
\end{theo}
\begin{proof}

This follows from combining the ideas from the proof of Lemma \ref{thm:multHardness} with previously known facts about $\BQP$. In particular, we will use the following facts:
\begin{enumerate}
\item There exists a uniform family of poly-size $\BQP$\footnote{Even $\IQP$ suffices here \cite{bremner2016average}.} circuits $C_x$ where $x\in\{0,1\}^n$ using a gate set with algebraic entries such that computing $|\bra{0^n}C_x\ket{0^n}|^2$ to multiplicative error $1/2$ is $\#\mathsf{P}$-hard \cite{bremner2016average}.
\item For any poly-sized quantum circuit $C$ over a gate set with algebraic entries, any non-zero output probability has magnitude at least inverse exponential \cite{kuperberg2009hard}.

\item As shown in the proof of Theorem \ref{cor:classificationWeak}, for any $U$ which is not a Clifford gate times a Z rotation, there is a postselection gadget $G$ which performs a unitary but non-Clifford one-qubit operation. Furthermore all ancilla qubits in $G$ begin in the state $\ket{0}$.
\end{enumerate}

From these facts, we can now prove the theorem. Let $p=|\braket{0^n|C_x }{0^n}|^2$. By Fact 2, the circuit $C_x$ from Fact 1 either has $p=0$ or $p \geq 2^{-O(n^c)}$ for some constant $c$.
Now suppose we compile the circuit $C_x$ from Fact 1 using Clifford gates plus the postselection gadget $G$ -- call this new circuit with postselection $C'_x$. By Sardharwalla \emph{et al.} \cite{sardharwalla2016universal} we can compile this circuit with accuracy $\epsilon=2^{-O(n^c)-100}$ with only polynomial overhead.

Let $\ell\in\{0,1\}^k$ be the string of postselection bits of the circuit $C'_x$ (which without loss of generality are the last bits of the circuit), and let $\alpha$ is the probability that all postselections succeed. Note $\alpha$ is a known and easily calculated quantity, since each postselection gadget is unitary so succeeds with a known constant probability.

Let $p' = |\braket{0^n \ell |C'_x}{0^{n+k}}|^2/\alpha$. Then we have that:
\begin{itemize}
\item If $p=0$ then $p' \leq 2^{-O(n^c)-100}$.
\item If $p \neq 0$ then  $p-2^{-O(n^c)-100} \leq p' \leq p + 2^{-O(n^c)-100} $. Since $p\geq2^{-O(n^c)}$, this is a multiplicative approximation to $p$ with error $ 2^{-100}$.
\end{itemize}

Now suppose that one can compute $|\braket{0^n \ell|C'_x}{0^{n+k}}|^2$ to multiplicative error $\gamma$ to be chosen shortly. Then immediately one can compute $p' =|\braket{0^n \ell |C'_x}{0^{n+k}}|^2/\alpha $ to the same amount of multiplicative error -- call this estimate $p''$. 
By the above argument, if $p=0$ then $p''<2^{-O(n^c)-100}(1+\gamma)$.
On the other hand if $p>0$ then $p'>2^{-O(n^c)}$, so $p''>2^{-O(n^c)}(1-\gamma)$. So long as $\gamma$ is chosen such that $2^{-100}(1+\gamma)<(1-\gamma)$ these two cases can be distinguished -- which holds in particular if $\gamma \approx 1/2$. 

Therefore, if $p''<2^{-O(n^c)}$ then we can infer that $p=0$.
If $p''>2^{-O(n^c)}(1-\gamma)$, then $p>0$ so $p''$ is a $\gamma$ approximation to $p'$ and hence a $\gamma+2^{-100}+\gamma2^{-100}$ approximation to $p$. In either case we have computed a $\gamma+2^{-100}+\gamma2^{-100}$ approximation to $p$. Therefore, if $\gamma=1/2 -2^{-99}$, then we have computed a 1/2-multiplicative approximation to $p$, which is $\#\mathsf{P}$-hard by Fact 1. Therefore, computing some the probability that the CCC correspoding to $C'_x$ outputs $\ket{0^{n}\ell}$ to multiplicative error $1/2 - 2^{-99}$ is $\#\mathsf{P}$-hard. One can similarly improve this hardness to $1/2 - o(1)$.

\end{proof}

Given that the worst-case version of Conjecture \ref{conj:avghardness} is true, a natural question to ask is how difficult it would be to prove the average-case conjecture. To do so would in particular prove quantum advantage over classical computation with realistic error, and merely assuming the polynomial hierarchy is infinite. In some ways this would be stronger evidence for quantum advantage over classical computation than Shor's factoring algorithm, as there are no known negative complexity-theoretic consequences if factoring is contained in $\mathsf{P}$. 

Unfortunately, recent work has shown that proving Conjecture \ref{conj:avghardness} would be a difficult task. Specifically, Aaronson and Chen \cite{aaronson2016complexity}  demonstrated an oracle relative to which $\PH$ is infinite, but classical computers can efficiently weakly simulate quantum devices to constant additive error. Therefore, any proof which establishes quantum advantage with additive error under the assumption that $\PH$ is infinite must be non-relativizing. In particular this implies any proof of Conjecture \ref{conj:avghardness} would require non-relativizing techniques -- in other words it could not remain true if one allows for classical oracle class in the circuit. This same barrier holds for proving the similar average-case hardness conjectures to show advantage for Boson Sampling, $\IQP$, $\DQC 1$, or Fourier sampling. Therefore any proof of Conjecture \ref{conj:avghardness} would require facts specific to the Clifford group. We leave this as an open problem. We also note that it remains open to prove the average-case \emph{exact} version of Conjecture \ref{conj:avghardness} - i.e. whether it is hard to exactly compute a large fraction of $p_{y,U,C}$. We believe this may be a more tractable problem to approach than Conjecture \ref{conj:avghardness}. However this remains open, as is the analogous average-case exact conjecture corresponding to $\IQP$. We note the corresponding average-case exact conjecture for Boson Sampling and Fourier sampling are known to be true \cite{aaronsonboson,feffermanfourier}, though these models are not known to anticoncentrate.

\section{Summary of simulability of CCCs}
\label{sec:simtypes}

For completeness, in this section we summarize the simulability of $U$-CCCs when $U$ is not a Clifford rotation times a $Z$ rotation. There are various notions of classical simulation at play here.
The results of this paper so far have focused of notions of approximate \emph{weak} simulation. A \textit{weak} simulation of a family of quantum circuits is a classical randomized algorithm that samples from the same distribution as the output distribution of the circuit. On the other hand, a \textit{strong} simulation of a family of quantum circuits is a classical algorithm that computes not only the joint probabilities, but also any marginal probabilities of the outcomes of the measurements in the circuit.
Following \cite{koh2015further}, we can further refine these definitions according to the number of qubits being measured: a strong(1) simulation computes the marginal output probabilities on individual qubits, and a strong($n$) simulation computes the probability of output strings $y\in\{0,1\}^n$. Similarly, a weak(1) simulation samples from the marginal output probabilities on individual qubits, and a weak($n$) simulation samples from $p(y_1,\ldots,y_n)$. A weak$^+$ simulation samples from the same distribution on all $n$ output qubits up to constant additive error. Our previous results have shown that efficient weak($n$) simulations (Theorem \ref{cor:classificationWeak}), weak$^+$ simulations (Theorem \ref{l1hardness}), and strong($n$) simulations (Theorem \ref{thm:worstcasehardness}) of CCCs are implausible.  
However it is natural to ask if it is possible to simulate single output probabilities of CCCs. 
It turns out the answer to this question is yes. This follows immediately from Theorem 5 of \cite{koh2015further}, which showed 
more generally that Clifford circuits with product inputs or measurements have an efficient strong(1) and weak(1) simulation. Therefore this completes the complexity classification of the simulability of such circuits. 
We note that $\IQP$ has identical properties in this regard. 
This emphasizes that the difficulty in simulating CCCs (or $\IQP$ circuits) comes from the difficulty of simulating all of the \emph{marginal} probability distributions contained in the output distribution, where the marginal is taken over a large number of output bits.
The probabilities of computing individual output bits of either model are easy for classical computation.
This is summarized in Figure \ref{fig:notionsOfxSimulation}.

\begin{figure}
\centering
\includegraphics[scale=0.5]{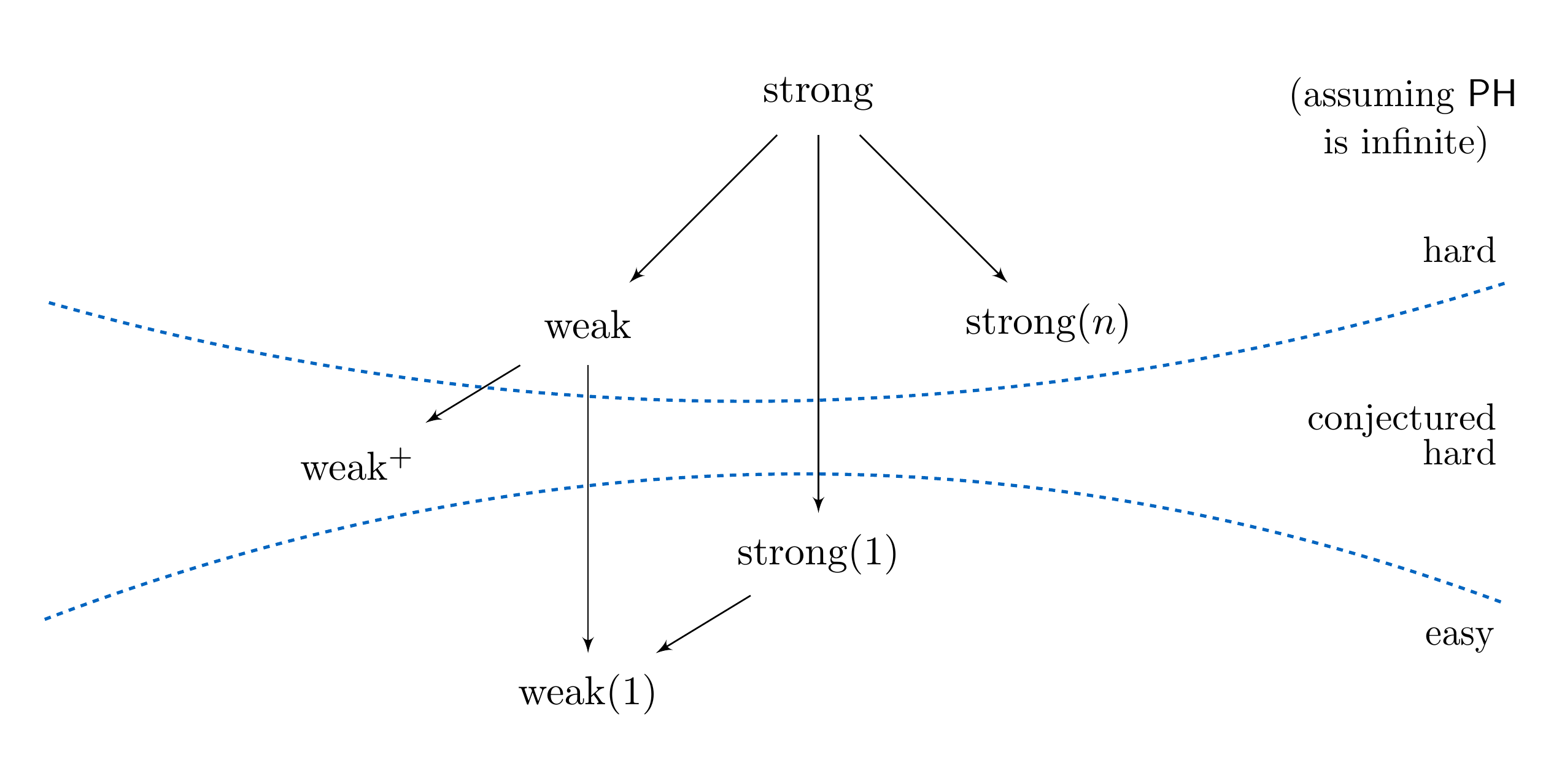}
\caption{\small Relationships between different notions of classical simulation and summary of the hardness of simulating CCCs. An arrow from $A$ to $B$ ($A \rightarrow B$) means that an efficient $A$-simulation of a computational task implies that there is an efficient $B$-simulation for the same task. Note also that an weak($n$) simulation exists if and only if a weak simulation exists. For a proof of these relationships, see \cite{koh2015further}. The two curves indicate the boundary between efficiencies of simulation of $U$-CCCs, where $U$ is not a Clifford operation times a $Z$ rotation. ``Hard'' means that an efficient simulation of $U$-CCCs is not possible, unless $\PH$ collapses. ``Conjectured hard'' means that an efficient simulation of $U$-CCCs is not possible, if we assume Conjecture \ref{conj:avghardness}. ``Easy'' means that an efficient simulation of $U$-CCCs exists. Note that when $U$ is a Clifford operation times a $Z$ rotation, all the above notions become easy. \newline}
\label{fig:notionsOfxSimulation}
\end{figure}

\section{Open Problems}
\label{sec:open}

Our work leaves open a number of problems.

\begin{itemize}

\item What is the computational complexity of commuting CCCs? In other words, can the gate set $CZ,S$ conjugated by a one-qubit gate $U$ ever give rise to quantum advantage? Note that this does not follow from Bremner, Jozsa and Shepherd's results \cite{BJS2010}, as their hardness proof uses the gate set $CZ,T$ or $CCZ,CZ,Z$ conjugated by one-qubit gates. If this is true, it would say that the ``intersection" of CCCs and IQP remains computationally hard. One can also consider the computational power of arbitrary fragments of the Clifford group, which were classified in \cite{cliffordclassification}. Perhaps by studying such fragments of the Clifford group one could achieve hardness with lower depth circuits (see additional question below).

\item We showed that Clifford circuits conjugated by tensor-product unitaries are difficult to simulate classically. A natural extension of this question is: suppose your gate set consists of all two-qubit Clifford gates, conjugated by a unitary $U$ which is \emph{not} a tensor product of the same one-qubit gate. Can one show that all such circuits are difficult to simulate classically (say exactly)? Such a theorem could be a useful step towards classifying the power of all two-qubit gate sets.

\item Generic Clifford circuits have a depth which is linear in the number of qubits \cite{aaronson2004clifford}. In particular the lowest-depth decomposition for a generic Clifford circuit over $n$ qubits to date has depth $14n-4$ \cite{maslov2017shorter}. Such depth will be difficult to achieve in near-term quantum devices without error-correction. As a result, others have considered quantum supremacy experiments with lower-depth circuits. For instance, Bremner, Shepherd and Montanaro showed advantage for a restricted version of $\IQP$ circuits with depth $O(\log n)$ \cite{bremner2016achieving} with long-range gates (which becomes depth $O(n^{1/2}\log n)$ if one uses SWAP gates to simulate long-range gates using local operations on a square lattice).  We leave open the problem of determining if quantum advantage can be achieved with CCCs of lower depth (say $O(n^{1/2})$ or $O(n^{1/3})$) with local gates only. 

\item In order to establish quantum supremacy for CCCs, we conjectured that it is $\#\mathsf{P}$-hard to approximate a large fraction of the output probabilities of randomly chosen CCCs (Conjecture \ref{conj:avghardness} ). Is it also $\#\mathsf{P}$-hard to exactly compute that large of a fraction of the output probabilities? This is a necessary but not sufficient condition for Conjecture \ref{conj:avghardness} to be true, and we believe it may be a more approachable problem.

\end{itemize}

\section*{Acknowledgments}
We thank Scott Aaronson, Mick Bremner, Alex Dalzell, Daniel Gottesman, Aram Harrow, Ashley Montanaro, Ted Yoder and Mithuna Yoganathan for helpful discussions. 
AB was partially supported by the NSF GRFP under Grant No. 1122374,  by a Vannevar Bush Fellowship from the US Department of Defense, and by an NSF Waterman award under grant number 1249349. 
 JFF acknowledges support from the Air Force Office of Scientific Research under AOARD grant no. FA2386-15-1-4082.
This material is based on research supported in part by the Singapore National Research Foundation under NRF Award No. NRF-NRFF2013-01.
DEK is supported by the National Science Scholarship from the Agency for Science, Technology and Research (A*STAR).

\bibliographystyle{alpha}
\bibliography{team} 

\appendix

\newpage
\section{Measurement-based Quantum Computing Proof of Multiplicative Hardness for CCCs for certain $U$'s}
\label{app:oldproof}

We will prove the following theorem using techniques from measurement-based quantum computing (MBQC). 

\begin{theo}
CCCs with $U=R_Z(\theta)H$ cannot be efficiently weakly classically simulated to multiplicative error $1/2$ unless the polynomial hierarchy collapses, for any $\theta$ which is not an integer multiple of $\pi/4$.
\end{theo}

This is a weaker version of Lemma \ref{thm:multHardness}. We include it for pedagogical reasons, as it provides a different way of understanding the main theorem using MBQC techniques, and it includes a more detailed walkthrough of the hardness construction. Furthermore, it does not rely on the theorem that the Clifford group plus any non-Clifford element is universal; instead one can directly prove postselected universality by finding a qubit rotation by an irrational multiple of $\pi$.

As in the proof of Lemma \ref{thm:multHardness}, we will first show that CCCs can perform universal quantum computation (i.e., the class $\BQP$) under postselection. This first step will make extensive use of ideas from Measurement-Based Quantum Computation \cite{briegel2009measurement}. Next, we will show that these circuits can furthermore perform $\PostBQP$ under postselection. This extension requires the inverse-free Solovay-Kitaev theorem of Sardharwalla \emph{et al.} \cite{sardharwalla2016universal}.

\begin{theo}
Postselected CCCs can be used to simulate universal quantum computation under the choice of $U=R_Z(\theta)H$, and for any choice of $\theta$ other than integer multiples of $\pi/4$.
\end{theo}

\begin{proof}

We will first describe the proof without reference to Measurement Based Quantum Computing (MBQC) so as to be understood by the broadest possible audience. We will then summarize the proof in MBQC language for those familiar with the area. 

Our proof will make use of four gadgets to show that under postselection, we can perform arbitrary 1-qubit gates in this model. For the first gadget, consider the following quantum circuit:

\begin{equation}\label{gadget:orig}\Qcircuit @C=1em @R=1em {
\lstick{\ket{\psi}} & \qw  & \ctrl{1}  & \gate{H} & \meter & \bra{0}   \\
\lstick{\ket{0}} & \gate{H} &  \control \qw  & \qw & \qw & \ket{\psi'} 
}
\end{equation}
Here the notation $\bra{0}$ denotes that we postselect that measurement outcome on obtaining the state $\ket{0}$, and the two-qubit gate is controlled-Z. This gadget performs teleportation \cite{gottesmanchuang1999}. One can easily calculate that $\ket{\psi'}=H\ket{\psi}$ -- in other words, this gadget performs the $H$ gate \cite{gottesmanchuang1999, BJS2010}. Likewise, if one postselects the first outcome to be $\ket{1}$, then the gate performed is $XH$. By chaining these gadgets together, one can perform any product of these operations. For instance, the following circuit performs $HXH$:
$$\Qcircuit @C=1em @R=1em {
\lstick{\ket{\psi}} & \qw  & \ctrl{1}  & \qw & \gate{H} & \meter & \bra{1}   \\
 \lstick{\ket{0}} & \gate{H} & \control \qw  &  \ctrl{1} & \gate{H} & \meter & \bra{0} \\
\lstick{\ket{0}} & \gate{H} & \qw  & \control \qw & \qw & \qw & \ket{\psi'} 
}
$$
The correctness follows from the fact that the order in which quantum measurements are taken is irrelevant. By stringing together $n$ of these, we can perform $n$ gates from the set $\{H, XH\}$. These generate a finite set of one-qubit gates which contain the Paulis. 

Now clearly circuits composed of these gadgets do not have the form of conjugated Clifford circuits with $U = R_Z(\theta) H$. But we can easily correct this by inserting $R_Z(\theta)$'s at the beginning of each line, and $R_Z(-\theta)$'s at the end of each line. 
\begin{equation}\label{Hgadgetfixed}\Qcircuit @C=1em @R=1em {
\lstick{\ket{\psi}} & \qw &  \gate{R_Z(\theta)} & \ctrl{1}   &   \gate{R_Z(-\theta)}& \gate{H} & \meter & \bra{0}   \\
\lstick{\ket{0}} & \gate{H} &  \gate{R_Z(\theta)}& \control\qw  & \gate{R_Z(-\theta)}  & \qw  & \qw & \ket{\psi'} 
}
\end{equation}
Clearly this is equivalent to our original gadget at the Z rotations commute through and cancel. Now the gadget has the property that
\begin{itemize}
\item Every input line begins with $R_Z(\theta)$, and every output line ends with $R_Z(-\theta)$. 
\item Every ancillary input begins with $\ket{0}$ then applies $R_Z(\theta)H$.
\item Every ancillary output applies $HR_Z(-\theta)$ and measures in the computational basis.
\item All gates in between are Clifford.
\end{itemize}
When composing such gadgets, the $R_Z(-\theta)$ at the end of each output line cancels with the $R_Z(\theta)$ at the beginning of each input line. Hence composing gadgets with the above properties will always form a CCC. For instance our prior circuit performing $HXH$ becomes
$$\Qcircuit @C=1em @R=1em {
\lstick{\ket{\psi}} &  \qw & \gate{R_Z(\theta)}  & \ctrl{1}  & \qw & \gate{R_Z(-\theta)}&\gate{H} & \meter & \bra{1}   \\
 \lstick{\ket{0}} & \gate{H} & \gate{R_Z(\theta)}& \control\qw  &  \ctrl{1} & \gate{R_Z(-\theta)}& \gate{H} & \meter & \bra{0} \\
\lstick{\ket{0}} & \gate{H} & \gate{R_Z(\theta)} & \qw  &\control\qw &\gate{R_Z(-\theta)}& \qw  & \qw & \ket{\psi'} 
}
$$

Thus, by simply replacing our input state $\ket{\psi}$ with the state $H\ket{0}$, and our output state with a Hadamard followed by measurement, this postselected circuit would be simulating the circuit which starts in the state $H\ket{0}$, applies $HXH$, then applies $H$ and measures. Furthermore, this state will have the form of a CCC. More generally, by stringing $n$ such gadgets together to form a CCC, clearly one can simulate any one-qubit quantum circuit where the initial state is $H\ket{0}$, one performs $n$ gates from the set $\{H,XH\}$, and then applies $H$ and measures. 

This allows us to simulate one-qubit gates from the set $\{H,XH\}$ with postselected CCC circuits. However, such gates are not universal for a single qubit. In order to show postselected CCCs can perform universal quantum computation, we will need to find a way to simulate all single qubit gates. To do so, we will consider adding features to our gadget.  So far the Clifford part of our CCCs are all commuting; let's consider adding a non-commuting one-qubit gate $X$ to make a new gadget:

\begin{equation}\label{xgadget}\Qcircuit @C=1em @R=1em {
\lstick{\ket{\psi}} & \qw &  \gate{R_Z(\theta)} & \ctrl{1}  &\gate{X} &   \gate{R_Z(-\theta)}& \gate{H} & \meter & \bra{0}   \\
\lstick{\ket{0}} & \gate{H} &  \gate{R_Z(\theta)}& \control\qw  &\qw & \gate{R_Z(-\theta)}  & \qw  & \qw & \ket{\psi'} 
}
\end{equation}
By commuting the $R_Z(\theta)$ rightwards on both lines, and noting that 
$$\Qcircuit @C=1em @R=1em {
& \qw &  \gate{R_Z(\theta)} & \gate{X} &   \gate{R_Z(-\theta)}&\qw 
}$$ 
is equivalent to
$$\Qcircuit @C=1em @R=1em {
& \qw &  \gate{R_Z(2\theta)} & \gate{X} &\qw 
}$$ 
we can see this performs the same quantum operation as 
\begin{equation*}\Qcircuit @C=1em @R=1em {
\lstick{\ket{\psi}} & \qw &  \qw & \ctrl{1}  &   \gate{R_Z(2\theta)} &\gate{X} & \gate{H} & \meter & \bra{0}   \\
\lstick{\ket{0}} & \gate{H} &  \qw& \control\qw  &\qw & \qw  & \qw  & \qw & \ket{\psi'} 
}
\end{equation*}
which since $HX=ZH$, is equivalent to
\begin{equation}\label{xgadgetrewritten}\Qcircuit @C=1em @R=1em {
\lstick{\ket{\psi}} & \qw &  \qw & \ctrl{1}  &   \gate{R_Z(2\theta)} &\gate{H} & \gate{Z} & \meter & \bra{0}   \\
\lstick{\ket{0}} & \gate{H} &  \qw& \control\qw  &\qw & \qw  & \qw  & \qw & \ket{\psi'} 
}
\end{equation}
By direct computation, gadget (\ref{xgadget}) (which is equivalent to gadget (\ref{xgadgetrewritten})) performs the operation $HR_Z(2\theta)$. Let us call this gate $G_0(\theta)$. Likewise, if one postselects on $\ket{1}$, one obtains the gate $G_1(\theta)=XHR_Z(2\theta)$. (This gadget is well-known in MBQC; see below). 

Therefore, by applying our gadgets (\ref{Hgadgetfixed}) and (\ref{xgadget}), we can create postselected CCCs to simulate the evolution of a one-qubit circuit which evolves by gates in the set $\{H,XH,G_0(\theta),G_1(\theta)\}$. Intuitively, as long as the choice of $\theta$ is not pathological, these gates will generate all one-qubit gates. Therefore we have all one-qubit gates at our disposal via these gadgets. We will prove this statement rigorously in Lemma \ref{lem:nobadtheta}, which we defer to the end of this appendix. In fact, we show that as long as $\theta$ is not set to $k\pi/4$ for some integer $k$, then the set of one qubit gates generated by these gadgets is universal on a qubit. 
Thus postselected CCC's (where $\theta\neq k\pi/4$) can simulate arbitrary one-qubit operations. 

To prove that postselected CCC's can perform universal quantum computation, we need to show how to perform an entangling two qubit gate. We can then appeal to the result of Brylinski \& Brylinski \cite{brylinski2002universal} and Bremner \emph{et al.} \cite{bremner2002practical} that any entangling two-qubit gate, plus the set of all-one-qubit gates, is universal for quantum computation. But performing entangling two-qubit gates is trivial in our setup, since the Clifford group (and the conjugated Clifford group) contains entangling two-qubit gates. For example, we can easily perform the controlled-Z gate between qubits with the following gadget:
\begin{equation}\label{czgadget}\Qcircuit @C=1em @R=1em {
 & \qw &  \gate{R_Z(\theta)} &\ctrl{1}  & \gate{R_Z(-\theta)} & \qw\\
& \qw & \gate{R_Z(\theta)} & \control\qw & \gate{R_Z(-\theta)} & \qw
}
\end{equation}

This gadget clearly has the correct form, and hence composes with the gadgets (\ref{Hgadgetfixed}) and \ref{xgadget} to form universal quantum circuits. This shows how to simulate $\BQP$ with postselected CCCs.

We can now recast this proof in the language of Measurement-Based Quantum Computing. Our result essentially follows from that fact that measuring graph states in the bases $HR_Z(2\theta)$ and $H$, combined with postselection, is universal for quantum computing. 
More formally, let $E$ be series of Controlled-Z operations that create a graph state out of $H^{\otimes n} \ket{0}^{\otimes n}$ (we will specify the cluster state later). Let $U=R_Z(\theta)H$ for some $\theta$ to be specified later. Then consider creating the CCC for the Clifford circuit $C=X^S E$, where the notation $X^S$ denotes that we apply an $X$ gate to some subset $S\subseteq [n]$ of the qubits. We have that
\begin{align*}
H^{\otimes n} R_Z(-\theta)^{\otimes n} X^{S} E R_Z(\theta)^{\otimes n} H^{\otimes n} \ket{0}^{\otimes n}
&= H^{\otimes n} R_Z(-\theta)^{\otimes n} X^{S}  R_Z(\theta)^{\otimes n} E H^{\otimes n} \ket{0}^{\otimes n}\\
&= H^{\otimes n} \left(XR_Z(2\theta)\right)^{S}  E H^{\otimes n} \ket{0}^{\otimes n} \\
&= H^{\otimes n} \left(XR_Z(2\theta)\right)^{S}  \ket{\mathrm{Cluster}} \\
&= \left(\left(Z H R_Z(2\theta)\right)^{S} \otimes H^{\bar{S}}\right)  \ket{\mathrm{Cluster}}  ,
\end{align*}
where the first equality follows from the fact that $R_Z$ and $E$ commute as they are both diagonal in the $Z$ basis, the second follows from the fact that on the lines without an $X$ the $R_Z(\theta) $ and the $R_Z(-\theta) $ cancel, and on the lines with an $X$ we have $R_Z(-\theta) X R_Z(\theta) = X R_Z(2\theta)$, the third follows from the fact that $E$ is constructed such that $E H^{\otimes n} \ket{0}^{\otimes n} = \ket{\mathrm{Cluster}}$, and the fourth from the fact that $HX=ZH$. Now since we're measuring in the Z basis at the end of the circuit, the last row of $Z$'s can be ignored, so the circuit is equivalent to:
\[ \left(\left(H R_Z(2\theta)\right)^{S} \otimes H^{\bar{S}}\right)  \ket{\mathrm{Cluster}} .  \]

Now we simply need to show that measurement based quantum computation with postselection on such a state is universal for quantum computing. In other words, we need to show that if we can construct a Cluster state and measure some qubits in the $H$ basis and others in the $HR_Z(2\theta)$ basis, and postselect on the outcomes, then we can perform universal quantum computation. It was previously known to be universal for MBQC if different $\theta$'s occur on each qubit \cite{briegel2009measurement}. In our setup we do not have this flexibility, but we instead have the additional ability to postselect.

Universality of this model follows from the fact that by preparing an appropriate Cluster state (using the standard trick to perform 1-qubit gates with MBQC), this gives us the ability to apply the one-qubit gate $H R_Z(2\theta)$ using postselection. Likewise, postselecting on $\ket{1}$ performs the operation $X H R_Z(2\theta)$. As discussed previously, by Lemma \ref{lem:nobadtheta}, as long as $\theta$ is not set to $k\pi/4$ for some integer $k$, this is a universal gate set on a qubit. The addition of entangling two-qubit operations on the Cluster state (namely, controlled-Z) boosts this model to universality.

\end{proof}

We have now shown that postselected CCCs can perform $\BQP$ under postselection. We now extend this to show they can perform $\PostBQP=\PP$ under postselection. This requires using the  inverse-free Solovay-Kitaev algorithm of \cite{sardharwalla2016universal}. From this, the hardness result follows via known techniques \cite{BJS2010, aaronsonboson}. 

\begin{theo}
Postselected CCCs with $U=R_Z(\theta)H$ can decide any language in $\PostBQP=\PP$, for any choice of $\theta$ other than integer multiples of $\pi/4$.
\end{theo}

\begin{proof}

To prove this, we will apply Aaronson's result that Postselected $\BQP$ circuits, denoted $\PostBQP$, can decide any language in $\PP$. Aaronson's proof works by showing that a particular universal quantum gate set -- namely the gate set $G$ consisting of Toffoli, controlled-Hadamard, and one qubit gates -- can decide $\PP$ under postselection. 

We previously showed that our postselected CCCs can perform a different universal quantum gate $G'$ consisting of controlled-Z, $HR_Z(2\theta)$, $XHR_Z(2\theta)$, $H$ and $XH$. Therefore, in order to show that postselected CCCs can compute $\PP$, we need to show how to simulate Aaronson's gate set $G$ using our gate set $G'$.

One difficulty is that we must be extremely accurate in our simulation of these gates. This is because postselected quantum circuits may postselect on exponentially tiny events. Therefore, in order to simulate Aaronson's postselected circuits for $\PP$, we will need to simulate each gate to inverse exponential accuracy. 

Normally in quantum computing this simulation is handled by the Solovay-Kitaev Theorem, which roughly states that any universal gate set can simulate any other universal gate set to error $\epsilon$ with only $\mathrm{polylog}(1/\epsilon)$ overhead. Therefore with polynomial overhead, one can obtain inverse exponential accuracy in the simulation. This is why the choice of gate set is irrelevant in the definition of $\PostBQP$. One catch, however, is that the Solovay-Kitaev theorem requires that the gate set is closed under inversion, i.e. for any gate $g\in G$, we have $g^{-1}\in G$ as well. This is an essential part of the construction of this theorem (which makes use of group commutators). It is an open problem to remove this requirement \cite{dawson2006solovay,kuperberg2009hard}. As a corollary, it is open whether or not the class $\PostBQP$ can still compute all languages in $\PP$ if the gate set used is not closed under inversion. It is possible the class could be weaker with non-inversion-closed gate sets.

Unfortunately, the gate set $G'$ we have at our disposal is not closed under inversion. Furthermore, since we obtained the gates using postselection gadgets, it is not clear how to generate the inverses of the gadgets, as postselection is a non-reversible operation. Therefore we cannot appeal to the Solovay-Kitaev theorem to show we can compute languages in $\PP$. 

Fortunately, however, even though our gate set does not have inverses, it does have a special property -- namely, our set of one qubit gates contains the Pauli group. It turns out that recently,  \cite{sardharwalla2016universal} proved a Solovay-Kitaev theorem for any set of one qubit gates containing the Paulis, but which is not necessarily closed under inversion. Therefore, by this result, even though our gate set is not closed under inversion, we can still apply any one-qubit gate to inverse exponential accuracy with merely polynomial overhead. So we can apply arbitrary one-qubit gates.

It turns out this is sufficient to apply gates from Aaronson's gate set $G$ consisting of Toffoli, controlled-H and one qubit gates with inverse exponential accuracy. To see this, first not that it is well-known one can construct controlled-V operations for arbitrary one-qubit gates V using a finite circuit of controlled-NOT and one-qubit gates -- see \cite{nielsenchuang} for details.  Furthermore, it is possible to construct Toffoli using a finite circuit of one qubit gates and controlled-V operations  \cite{nielsenchuang}. This, together with the fact that controlled-NOT is equal to controlled-Z conjugated by Hadamard on one qubit, shows that each gate in $G$ has an exact decomposition as a finite number of controlled-Z gates and one-qubit gates. Hence, using controlled-Z gates and one-qubit gates compiled to exponential accuracy, one can obtain circuits from $G$ with inverse exponential accuracy. Thus, our gate set $G'$ can efficiently simulate gates from $G$, and hence our postselected CCCs can compute all languages in $\PostBQP=\PP$ as well.

\end{proof}

From this, the hardness result follows via known techniques \cite{BJS2010, aaronsonboson}. 

\begin{cor}
Conjugated Clifford circuits cannot be weakly simulated classically to multiplicative error unless the polynomial hierarchy collapses to the third level, for the choice of $U=R_Z(\theta)H$ for any $\theta$ which is not an integer multiple of $\pi/4$.
\end{cor}

To complete our proof, we merely need to show the following lemma:

\begin{lemma}\label{lem:nobadtheta} So long as $\theta$ is not an integer multiple of $\pi/4$, then the gates $H R_Z(2\theta)$ and $X H R_Z(2\theta)$ are universal on a qubit. Furthermore, so long as $\theta$ is not an integer multiple of $\pi/4$, then at least one of these gates is a rotation of the Bloch sphere by an irrational multiple of $\pi$.
\end{lemma}
\begin{proof}
For convenience of notation, define $G_0 = -iH R_Z(2\theta)$ and $G_1= -X H R_Z(2\theta)$. We will actually begin by proving something stronger: namely, that as long as $\theta$ is not an integer multiple of $\pi/4$, then one of the rotations $G_0$ and $G_1$ is by an irrational multiple of $\pi$.

We will prove this by contradiction. Suppose that both $G_0$ and $G_1$ are rotations by rational multiples of $\pi$, call their rotation angles $\phi_0$ and $\phi_1$, respectively. By direct computation, first eigenvalue of $G_0$ is given by 
\[\frac{1}{2\sqrt{2}} \left(  -2\sin (\theta) - i\sqrt{6+2\cos(2\theta)} \right). \]
Since this must be equal to $e^{\pm i\phi_0/2}$, and by considering the real part of this equation, we have that
\begin{equation}\label{eq:g0angle}
\cos(\phi_0/2) = -\frac{\sin (\theta)}{\sqrt{2}}.
\end{equation}
By an identical argument, for gate $G_1$ we have that
\begin{equation}\label{eq:g1angle}
\cos(\phi_1/2) = - \frac{\cos (\theta)}{\sqrt{2}}.
\end{equation}
Squaring these terms and summing them, we obtain that
\[\cos^2(\phi_0/2) +\cos^2(\phi_1/2) = \frac{1}{2}. \]
Or, applying the fact $\cos^2 t = \frac{1+\cos 2t }{2}$ and simplifying, one can see this is equivalent to
\[\cos (\phi_0) + \cos(\phi_1) +\cos(0)= 0. \]
Since we are assuming by way of contradiction that $G_0,G_1$ are of finite order, we are assuming that $\phi_0,\phi_1$ are rational multiples of $\pi$. Previously, Crosby \cite{crosby1946} and W\l{}odarski \cite{wlodarski} classified all possible solutions to the equation $\cos(\alpha_1)+\cos(\alpha_2)+\cos(\alpha_3)$ where each $\alpha_i$ are rational multiples of $\pi$. The four possible solution families to this equation (assuming without loss of generality that $0\leq\alpha_i\leq\pi$) are \cite{wlodarski}
\begin{itemize}
\item $\left\{\beta,\pi-\beta,\pi/2\right\}$ where $0\leq\beta\leq\pi$
\item $\left\{\delta, \frac{2\pi}{3}-\delta,\frac{2\pi}{3}+\delta \right\}$ where $0\leq\delta\leq \frac{\pi}{3} $
\item $\left\{\frac{2\pi}{5},\frac{4\pi}{5},\frac{\pi}{3}\right\}$
\item $\left\{\frac{\pi}{5},\frac{3\pi}{5},\frac{2\pi}{3}\right\} .$
\end{itemize}

Since we have that one of our three angles is $0$, the latter two cases are immediately ruled out, and we must have that the angles $\{\phi_0,\phi_1\}$ are either $\{\pi/2,\pi\}$ or $\{2\pi/3,2\pi/3\}$. One can easily see that the first solution corresponds to $\theta=k\pi/2$ for an integer $k$, and the second solution corresponds to $\theta = k\pi/4$ for an odd integer $k$.  

Therefore, so long as $\theta$ is not an integer multiple of $\pi/4$, we have a contradiction, as there are no further solutions to these equations where the $\phi_i$ are rational multiples of $\pi$. So if $\theta$ is set to any value other than $k\pi/4$ for an integer $k$, we have that at least one of the gates $G_0$ and $G_1$ is a rotation by an irrational multiple of $\pi$.

Now what remains to be shown is that the gates $G_0$ and $G_1$ are universal in the general case. This can be shown easily by the classification of continuous subgroups of $SU(2)$. The continuous subgroups of $SU(2)$ are $U(1)$ (corresponding to all rotations about one axis), $U(1)\times \mathbb{Z}_2$ (corresponding to all rotations about an axis $a$, plus a rotation by $\pi$ through another axis perpendicular to $a$), and $SU(2)$.
By our prior result we know that either $G_0$ or $G_1$ generates all rotations about its axis of rotation on the Bloch sphere. Therefore, if we can show that neither $G_0$ nor $G_1$ are rotations by angle $\pi$ we are done, as these then must generate all of $SU(2)$. However this follows immediately from equations \ref{eq:g0angle} and \ref{eq:g1angle}, since these equations imply that we can have either $\phi_0=\pi$ or $\phi_1=\pi$ only when $\theta$ is a rational multiple of $\pi/2$. Hence, as long as $\theta$ is not a rational multiple of $\pi/4$, neither $G_0$ nor $G_1$ is a rotation by $\pi$, and furthermore one is a rotation by an irrational multiple of $\pi$. These gates generate a continuous group which is neither $U(1)$ nor $U(1)\times \mathbb{Z}_2$, and therefore by the above observation these generate all of $SU(2)$.

\end{proof}

\end{document}